%% file: main.tex
\documentclass[11pt]{article}
\input{0_document_setup}

\begin{document}
\title{Catch me if you can: \\ Signal localization with knockoff e-values}
\author[1]{Paula Gablenz}
\author[2,1]{Chiara Sabatti}
\affil[1]{Department of Statistics, Stanford University}
\affil[2]{Department of Biomedical Data Science, Stanford University}

\date{}

\maketitle

\begin{abstract}
We consider problems where many, somewhat redundant, hypotheses are tested and we are interested in reporting the most precise rejections, with false discovery rate (FDR) control. This is the case, for example, when researchers are interested both in individual hypotheses as well as group hypotheses corresponding to intersections of sets of the original hypotheses, at several resolution levels. A concrete application is in genome-wide association studies, where, depending on the signal strengths, it might be possible to resolve the influence of individual genetic variants on a phenotype with greater or lower precision. To adapt to the unknown signal strength, analyses are conducted at multiple resolutions and researchers are most interested in the more precise discoveries. Assuring FDR control on the reported findings with these adaptive searches is, however, often impossible. To design a multiple comparison procedure that allows for an adaptive choice of resolution with FDR control, we leverage e-values and linear programming. We adapt this approach to problems where knockoffs and group knockoffs have been successfully applied to test conditional independence hypotheses. We demonstrate its efficacy by analyzing data from the UK Biobank.
\end{abstract}


\section{Introduction}\label{sec:introduction}

\input{01_introduction.tex}

\section{Problem statement and notation}\label{sec:notation}
\input{02_notation}

\section{eLP: an e-value based linear program for controlled, resolution-adaptive discoveries}\label{sec:elp}
\input{03_elp}

\section{KeLP: eLP with knockoff e-values}\label{sec:kelp}
\input{04_kelp}

\section{Additional applications of knockoff e-values and KeLP}\label{sec:other_knock_eval_applications}
\input{05_other_knockoff_eval_applications}

\section{Simulations}\label{sec:simulations}
\input{06i_simulation}

\section{Application to the UK Biobank}\label{sec:application_ukb_height}
\input{06ii_UKB_application}

\section{Conclusion}\label{sec:conclusion}

\input{07_conclusion}

\clearpage
\addcontentsline{toc}{section}{References}
\bibliographystyle{rss} 
\normalsize 
\bibliography{lit}

\newpage
\setcounter{page}{1}
\appendix
\setcounter{equation}{0}
\renewcommand\theequation{S.\arabic{equation}}

\setcounter{theorem}{0}
\renewcommand{\thetheorem}{S.\arabic{theorem}}

\renewcommand{\thealgocf}{S.\arabic{algocf}}

\setcounter{ProcedureDef}{0}
\renewcommand{\theProcedureDef}{S.\arabic{ProcedureDef}}

\begin{appendices}

\section{Proof of Theorem \ref{thm:main_thm_elp}}\label{appendix:proofs_main}
\input{08i_appendix_proof_main}

\section{Focused e-BH vs eLP}\label{appendix:focusedeBH_vs_kelp}

\input{08ii_appendix_focusedeBH_vs_elp}

\section{Partial conjunctions}\label{appendix:global_partial_description}
\input{08iii_appendix_global_partial}

\section{e-Filter and e-MKF}\label{appendix:emlkf_description}
\input{08iv_appendix_emlkf}

\section{Details on the simulations}\label{appendix:details_simulation}

\input{08v_details_simulation}

\section{Details on the UKB application}\label{appendix:details_ukb}

\input{08vi_details_ukb_application}


\end{appendices}

\end{document}

%% file: 0_document_setup.tex
\usepackage{color,colortbl}
\usepackage[dvipsnames]{xcolor}
\usepackage{amsmath,amsthm}
\usepackage{amsfonts}
\usepackage{lscape}
\usepackage{graphicx}
\usepackage{pdflscape}
\usepackage{pgf,pgffor}
\usepackage{epstopdf}
\graphicspath{{../}}
\usepackage{natbib}
\usepackage[english]{babel}
\usepackage{float}
\usepackage{standalone}
\usepackage{siunitx}
\usepackage{hyperref}
\usepackage{threeparttable}
\usepackage{threeparttablex}
\usepackage{url}
\usepackage{graphicx}
\usepackage{enumerate}  
\usepackage{xcolor}
\usepackage{tikz}
\usepackage[ruled]{algorithm2e}
\usepackage{lscape}
\usepackage{enumitem}
\hypersetup{%
 citecolor=blue,
  colorlinks = true,
  linkcolor  = blue, 
  urlcolor = black
}
\let\oldcite=\cite
\let\oldcitet=\citet
 
\renewcommand{\citet}[2][]{\textcolor{blue}{\oldcitet[#1]{#2}}}
\renewcommand{\cite}[2][]{\textcolor{blue}{\oldcite[#1]{#2}}}

\usepackage{authblk}
\usepackage{setspace}
\usepackage{lscape}
\usepackage{amssymb}   
\usepackage{setspace}
\usepackage{multirow}
\usepackage{longtable}
\usepackage{rotating}
\usepackage{appendix}
\usepackage{geometry}
\usepackage{adjustbox}
\newlength\myHeight 
\newlength\myWidth
\usepackage{verbatim}
\usepackage{siunitx}
\usepackage{pgf,pgffor}
\usepackage{threeparttable}
\usepackage{cprotect}
\usepackage[T1]{fontenc}
\usepackage{comment}
\usepackage[parfill]{parskip}
\usepackage{caption}

\usepackage{amsmath}
\usepackage{amssymb}
\usepackage{amsthm}
\newtheorem{theorem}{Theorem}

\newtheorem{lemma}[theorem]{Lemma}

\newtheorem{ProcedureDef}{Procedure}

\newtheorem*{goal*}{Goal}

\usetikzlibrary{patterns}

\newcommand{\indep}{\perp \!\!\! \perp}

\newcommand{\scriptS}[1]{\mathcal{S}}
\newcommand{\scriptF}[1]{\mathcal{F}}

\makeatletter
\newcommand\dotover{\leavevmode\cleaders\hb@xt@ .22em{\hss $\cdot$\hss}\hfill\kern\z@}

\makeatother

\usepackage{titlesec}
\SetKwInput{KwIn}{Input}
\SetKwInput{KwOut}{Output}

\input{tikzlibrarybayesnet.code}

%% file: 01_introduction.tex
Finding, among many features, those that contain information on an outcome of interest is an important statistical problem which can be described in terms of testing multiple conditional independence hypotheses. For example, genome wide association studies (GWAS) are devoted to identifying single nucleotide polymorphisms (SNPs) whose alleles influence a medically relevant phenotype $Y$. A discovered SNP is truly interesting when it provides information on $Y$ in addition to that available in the rest of the genome. For each polymorphism $j$ in the study, a conditional independence hypothesis can be written as 
\begin{equation}\label{eq:intro_example_hypothesis}
 H_j: Y \indep \text{SNP}_j \mid \text{SNP}_{-j} 
\end{equation}
with $j \in \{1, ..., p\}$ and where $\text{SNP}_{-j}$ denotes all SNPs except SNP $j$.
The process with which we inherit genetic material from our parents is such that, within a human population, the random variables describing alleles at SNPs located in proximity of each other along a chromosome are dependent. This local dependency, known as linkage disequilibrium, can make it difficult to reject hypotheses of the type of (\ref{eq:intro_example_hypothesis}), as the variation in each SNP is typically well recapitulated by its neighbors. To avoid power loss, it is convenient to test conditional hypotheses corresponding to groups of highly correlated SNPs. This is, for example, the solution adopted in \citet{sesia2020multi}, where hierarchical clustering is used to define multiple groupings of SNPs, corresponding to different correlation thresholds: conditional testing with false discovery rate (FDR) control is carried out for each of these partitions in the genome, resulting in discoveries at multiple resolutions, as the one illustrated in the top panel of Figure \ref{fig:KnockoffZoom}.

\begin{figure}[H]
 \centering
 \includegraphics[width = \textwidth]{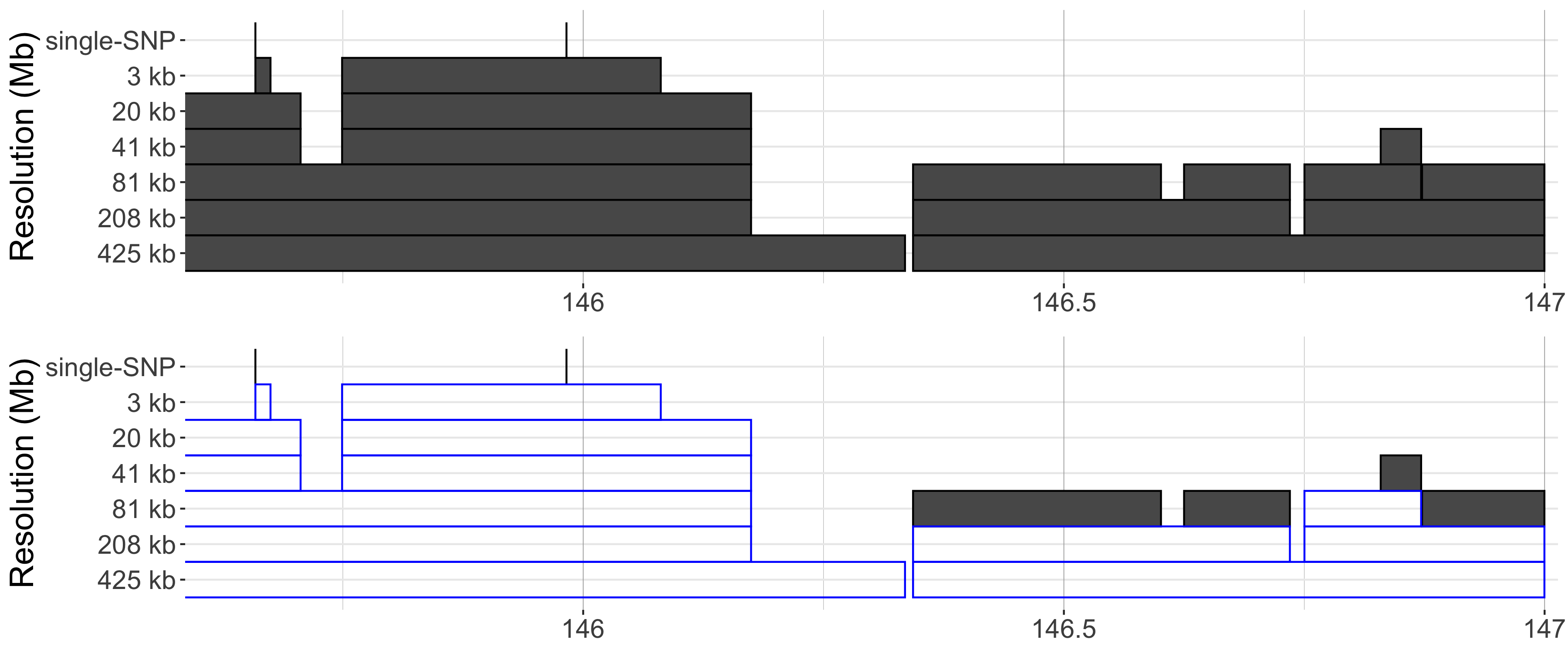}
 \caption{
Examples of the results of the analysis of height in the UK Biobank \citep{sesia2020multi}. Genomic position is on the $x$-axis (this is a small portion of chromosome 4). Different heights on the $y$-axis correspond to different resolutions. Top panel: Each filled gray box corresponds to a group of SNPs for which the hypothesis of independence from height given the rest of the genome has been rejected, carrying out an analysis that controls FDR at each resolution \citep{sesia2020multi}. Lower panel: the most specific discoveries among those in the top panel are highlighted as filled gray boxes, while the others (logically implied) are empty rectangles.}
 \label{fig:KnockoffZoom}
\end{figure}

Scientists, presented with these results, are interested in focusing on the most precise discoveries, highlighted in the bottom panel of Figure \ref{fig:KnockoffZoom}. However, such post-hoc selection of non redundant findings breaks FDR guarantees provided by the original algorithm \citep{katsevich2021filtering}. This limitation of FDR procedures is well known \citep{GS2011} and is one reason that makes alternative approaches more attractive, despite possible loss of power (for example, in the context of genetic analysis; see the contribution by \citet{Ret2020} and related discussion papers). Another line of work for controlling false discoveries by providing upper confidence bounds on the number of false discoveries has been developed by \citet{genovese2004stochastic, GS2011, goeman2019simultaneous} and extended to e-values by \citet{vovk2023confidence}.

GWAS are but one example of what one might call ``resolution adaptive'' testing: problems where scientists are aiming to simultaneously discover a signal and localize it as precisely as possible \citep{spector2023controlled}. Localization might refer to a physical position in the genome or in space \citep{Rosenblatt2018}, but it can also more broadly be construed as precision of the finding: in \citet{cortes2017bayesian}, for example, localizing a signal corresponds to identifying the nodes farther from the root in a tree describing a set of related diseases. 
There is typically a tension between the precision of a discovery and the ability to detect the underlying signal: scientists are interested in finding the most high resolution at which interesting rejections can be made. When the space of hypotheses explored is very large, it is reasonable to expect that this ``optimal'' resolution might vary across the existing signals. At the same time, the adaptive selection of resolution is complicated by the need to account for multiplicity, providing some global error guarantees. 

The literature documents multiple approaches to tackling this problem. Without attempting a comprehensive review, it it worth pointing out that they vary with respect to target error rate (ex. FWER or FDR), and the assumptions they make on the relation between hypotheses at different levels of resolution. For example, 
\citet{meijer2015multiple} develop a sequential multiple testing method to control the FWER for hypotheses that can be described by a directed acyclic graph and they are able to increase power by leveraging logical relations between hypotheses. A similar setting is considered in \citet{ramdas2017dagger}, which targets FDR control. In this work, we consider settings where no restriction is placed on the hypotheses at different resolutions (for example, we do not require partitions to be nested). This allows us to model situations where the different resolutions capture different, possibly quite unrelated, procedures to query the data.

\citet{spector2023controlled} recently studied this problem. Following \citet{mandozzi2016hierarchical}, they provide a clear formalization of the tension between localization and discovery and give an elegant and efficient solution in a Bayesian framework (BLiP). While guarantees on the Bayesian FDR are a useful step forward, when posterior probabilities for each of the tested hypotheses are available, they are not satisfactory for a large portion of the scientific community who finds the frequentist approach to inference an easier platform for agreement. 

Working in a frequentist framework, \citet{katsevich2021filtering} considered a larger class of multiple testing problems where researchers might be interested in refining the set of discoveries to eliminate redundancy, broadly defined. They introduce a multiple testing procedure (Focused BH) that generalizes the Benjamini and Hochberg (BH) procedure \citep{BH95}, and controls the FDR when the p-values for all hypotheses considered are independent or have a special type of positive dependence called PRDS (Positive Regression Dependency on each one from a Subset) \citep{BY2001}. Given the complex dependence structure of tests build on overlapping/nested sets of predictors as the ones in GWAS, these distributional assumptions are limiting, and the extension of Focused BH described in \citet{katsevich2021filtering} to arbitrary filters and dependency structures is very conservative.

An example of special interest is the multi resolution analysis of GWAS data with conditional hypotheses \citep{sesia2020multi} enabled by the knockoff framework \citep{barber2015controlling, candes2018panning} illustrated in Figure \ref{fig:KnockoffZoom}. Here the p-values are dependent and they are designed to be analyzed with a special multiple comparison adjustment procedure that is not captured by Focused BH. 

To address this gap, we introduce a resolution-adaptive, multiple comparison procedure that leads to frequentist FDR control without making distributional assumptions on the test statistics. A key ingredient to our approach are e-values, and their ``easy'' calculus \citep{vovk2021valuescalibrationcombination, wang2022eBH}. Working with e-values, we can define a linear program related to that one in \citet{spector2023controlled} and leading to frequentist FDR control. 

These results are concretely interesting because we can describe powerful e-values for resolution adaptive variable selection. \citet{ren2023derandomized} have recently shown how the knockoff filter can be re-interpreted as a BH-type procedure \citep{wang2022eBH} on specially constructed e-values. Building on their work, we introduce KeLP (Knockoff e-value Linear Program), with which we can analyze genomescale data in a powerful and time-efficient manner.

The rest of the paper is organized as follows. We formally introduce the problem in section \ref{sec:notation}, describe a solution using e-values in section \ref{sec:elp}, and operationalize it for testing conditional hypotheses with knockoffs in section \ref{sec:kelp}. We describe additional applications of KeLP in section \ref{sec:other_knock_eval_applications}. Section \ref{sec:simulations} is devoted to simulations and Section \ref{sec:application_ukb_height} presents the results of applying our methods to the UK Biobank data.

%% file: 02_notation.tex
We consider problems where investigators are interested in evaluating $p$ null hypotheses $\{H_1, ...., H_p\}$ as well as ``group'' hypotheses, which correspond to intersections of the original ones. An individual hypothesis can contribute to the definition of multiple group hypotheses, corresponding to different levels of resolution. 
For example, functional magnetic resonance imaging (fMRI) studies measure millions of voxels at different time intervals. While the most precise hypotheses that can be investigated describe the behavior of one voxel at one time point, scientists are also interested in aggregating the signal to correspond to larger structures and time frames. 

Formally, let $\mathcal{M}$ denote a set of resolutions and $\mathcal{A}^m$ denote a partition of $\{1, \ldots,p\}$ into disjoint groups at resolution $m \in \mathcal{M}$. With $A_{g}^m \in \mathcal{A}^m$ for $g \in \{1, ..., |\mathcal{A}^m|\}$ we indicate the elements of the partition $\mathcal{A}^m$, each of which is a disjoint set of indices $j \in \{1, ..., p\}$, which we refer to as a group. 

Each group defines a hypothesis
\begin{equation}\label{eq:group_hypothesis}
H_g^m=\bigcap_{j \in A_g^m} H_j \quad \text { for all } g, m.
\end{equation}

To avoid confusion, we use $m=1$ to indicate the partition $\mathcal{A}^1$ of $\{1, \ldots,p\}$ into $p$ sets of size 1, so that $H_j^1=H_j$ for $j=1,\ldots, p$. Referring to the index $m$ as a resolution, we underscore how the various partitions have a different degree of coarseness; however, we do not require that the partitions are nested, nor that there is a natural ordering among them. Let $\mathcal{H}^m$ denote the set of resolution-specific hypotheses for $m \in \mathcal{M}$ and $\mathcal{H}_0^m$ the corresponding set of true nulls. Combining several levels of resolution, let $\mathcal{H} = \{\mathcal{H}^m\}_{m \in \mathcal{M}}$ denote the combined set of hypotheses across resolutions and $\mathcal{H}_0 = \{ \mathcal{H}_0^m \}_{m \in \mathcal{M}}$ the combined set of true null hypotheses. We also let $\mathcal{A} = \{\mathcal{A}^m\}_{m \in \mathcal{M}}$ denote the set of groups across all levels of resolution. For simplicity of notation, we will let a group in $\mathcal{A}$ be denoted as $A \in \mathcal{A}$ with the understanding that these groups are of course resolution-specific. To each group $A \in \mathcal{A}$ corresponds an hypothesis $H_A \in \mathcal{H}$, for which we have a valid test (as described further in sections \ref{sec:elp} and \ref{sec:kelp}, we will work with e-values). Figure \ref{multires} gives a schematic illustration of a multi-resolution family of hypotheses $\mathcal{H}$.

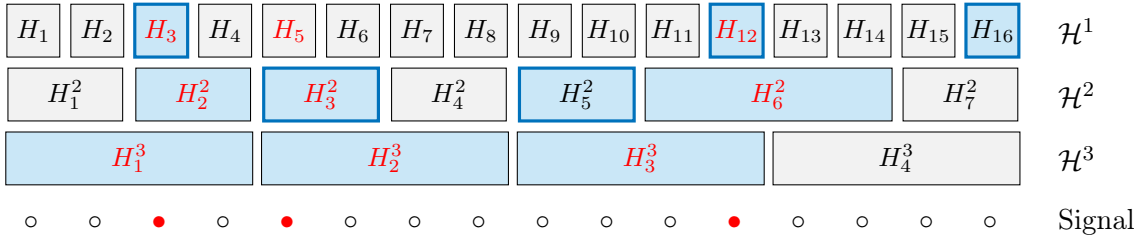
\begin{figure}
 \begin{center}

\input{figures/cartoon_figures/schematic_representation_multi_resolution.tex}
\caption{Schematic representation of a multi-resolution family of hypotheses $\mathcal{H}=\{\mathcal{H}^1,\mathcal{H}^2,\mathcal{H}^3\}$ and a rejection set ${\mathcal R}$. Resolution 1 corresponds to individual level hypotheses; resolutions 2 and 3 correspond to two different partitions of the 16 hypotheses. False null hypotheses are indicated in red and capture the underlying signal at different resolutions. The hypotheses part of the rejection set $\mathcal R$ are shaded in blue; the most specific rejections are highlighted with thick blue contours. The overall FDP is 2/10. The FDP for resolution 1 is 1/3; for resolution 2 the FDP is 1/4, and for resolution 3, the FDP is 0. The FDP for the set of non-redundant and most precise discoveries is 2/5.
}\label{multires}
\end{figure}

When the size of these families is large, it is necessary to control some measure of global error and in many scientific settings the false discovery rate \citep{BH95} is often a preferred choice. Given the structure of the multi-resolution families, however, there are multiple ways of defining the FDR: one can consider the total collection of rejections across resolutions, and control the expected value of the false discovery proportion among them; or one can control the FDR within each resolution separately - possibly enforcing consistency across resolutions; or one might want to focus on the most precise discoveries made (the ones indicated with a blue thick contour in Figure \ref{multires}) and control the FDR among them. 

All of these choices are documented in the literature. For example, as anticipated, KnockoffZoom \citep{sesia2020multi} controls the FDR separately within each resolution; the p-filter \citep{barber2017pfilter} and the Multilayer Knockoff Filter (MKF) \citep{katsevich2019multilayer} result in coordinated discoveries across resolutions, with simultaneous FDR control within each resolution. The most precise discoveries have been described as ``outer-nodes'' when the partitions corresponding to different resolutions are nested and the tested hypotheses can be organized in a hierarchical tree structure: \citet{yekutieli2008hierarchical} proves the first result of FDR control in this framework, under some restrictive assumptions. \citet{mandozzi2016hierarchical} tackle the problem from the point of view of Family Wise Error Rate (FWER). Both BLiP \citep{spector2023controlled} and Focused BH \citep{katsevich2021filtering} consider approaches to identify non-redundant and precise rejection sets and develop methodology that leads to FDR control in some settings. This is also our goal.

Specifically, 
let $\mathcal{F}$ be a multiple testing procedure applied to $\mathcal{H}$ that outputs a rejection set $\mathcal{R}_{\mathcal{F}}$. Let ${\mathcal A}_{{\mathcal R}_{\mathcal{F}}} =
\{A \;\;s.t. \;\;H_A\in {\mathcal R}_{\mathcal{F}}\}$. We define the rejections $\mathcal{R}_{\mathcal{F}}$ to be non-redundant if ${\mathcal A}_{{\mathcal R}_{\mathcal{F}}}$ contains disjoint sets. As in \citet{mandozzi2016hierarchical}, we are interested in ``minimal'' non-redundant discoveries, that is those corresponding to the groups with the smallest cardinality; at the same time, as in \citet{spector2023controlled}, we do not want to restrict ourselves to nested partitions. To capture the value of rejecting a particular hypothesis $H_A$, we use weights $w(A)$, as in \citet{spector2023controlled}; in particular we often focus on the original suggestion in \citet{mandozzi2016hierarchical} with $w(A)=1/|A|,$ where $|A|$ is the cardinality of $A$. The following formalizes our goal of finding a multiple testing procedure which maximizes the resolution-adjusted rejection count, controlling the FDR and leading to non-redundant findings.

\begin{goal*}{\em \textbf{Resolution-adaptive discovery with FDR control.}}\label{def:goal} Given a multi-resolution family of hypotheses $\mathcal H$ as described above, corresponding to a collection of groups $\mathcal A$; 
a fixed weighting function $w: \mathcal{A} \to \mathbb{R}^+$, capturing the value of discoveries at different resolutions; and a level $\alpha \in (0, 1)$ of desired FDR, we seek a multiple testing procedure $\cal F$ which solves the following constrained optimization problem:
 \begin{eqnarray}
\max & \displaystyle \sum_{A\in {\mathcal A}_{{\mathcal R}_{\mathcal{F}}}}w(A) \label{eq:AdjCount} \\
\text { s.t. } & \text{FDR}(\mathcal{F}) = \mathbb{E} \Big[ \displaystyle \frac{|\mathcal{R}_{\mathcal{F}} \cap \mathcal{H}_0|}{1 \vee |\mathcal{R}_{\mathcal{F}}|}\Big] \leq \alpha \label{eq:fdr}\\
& A, B \in {\mathcal A}_{{\mathcal R}_{\mathcal{F}}} \Rightarrow A\cap B =\varnothing \label{eq:Disj}
\end{eqnarray}
\end{goal*}

Before introducing a procedure that solves this problem, we pause to make a few remarks. First, in contrast to \citet{spector2023controlled}, who use posterior probabilities to maximize a notion of expected power, our objective (\ref{eq:AdjCount}) is to maximize the number of (weighted) rejections: we want the largest number of  valuable discoveries, while controlling FDR.

Second, while we focus on situations where the value of discoveries is described by weights $w(A)$ decreasing in $|A|$, in different contexts it might be useful to consider other evaluations: the procedures we will describe adapt to any, as long as the weights are fixed.

The constraint (\ref{eq:Disj}) expresses our interest in obtaining non-redundant rejections.  In section \ref{subsec:emlkf} and appendix \ref{appendix:emlkf_description} we will discuss, instead, procedures that lead to discoveries at multiple resolutions, in a coordinated fashion.

We also want to underscore how the notion of resolution-adaptive discovery is relevant also for families of hypotheses that are not associated to spatial localization of a signal. For example, as in \citet{cortes2017bayesian}, low resolution hypotheses might correspond to group of traits, and the precision of discoveries increases as more specific statements can be made about which traits are involved. We will consider related problems in sections \ref{subsec:structured_outcomes_description} and \ref{subsec:structured_outcomes_simulation}.

Finally, note that it is also possible to generalize the above definition of $\mathcal{A}^1$ to sets containing multiple ``individual'' hypotheses, as long as these are always tested together.

%% file: figures/cartoon_figures/schematic_representation_multi_resolution.tex
  \begin{tikzpicture}[scale=.85]
 
  \node[obs, rectangle, text=gray, fill=gray!10]  (H1)  at(0,4)   {\color{black} $H_{1}$} ; 
   \node[obs,  rectangle, text=gray, fill=gray!10]  (H2)  at(1,4)   {\color{black} $H_{2}$} ; 
     \node[obs,  rectangle,text=gray, fill=Cerulean!20,very thick,draw=RoyalBlue]   (H3)  at(2,4)   {\color{red} $H_{3}$} ; 
   \node[obs, rectangle, text=gray, fill=gray!10]  (H4)  at(3,4)   {\color{black} $H_{4}$} ; 
     \node[obs, rectangle,  text=gray, fill=gray!10]  (H5)  at(4,4)   {\color{red}$H_{5}$} ; 
   \node[obs, rectangle, text=gray, fill=gray!10]  (H6)  at(5,4)   {\color{black} $H_{6}$} ; 
   \node[obs,  rectangle, text=gray, fill=gray!10]  (H7)  at(6,4)   {\color{black} $H_{7}$} ; 
   \node[obs, rectangle, text=gray, fill=gray!10]  (H8)  at(7,4)   {\color{black} $H_{8}$} ; 
     \node[obs,  rectangle,text=gray, fill=gray!10]  (H9)  at(8,4)   {\color{black} $H_{9}$} ; 
   \node[obs, rectangle,  text=gray, fill=gray!10]  (H10)  at(9,4)   {\color{black} $H_{10}$} ; 
     \node[obs,  rectangle,text=gray, fill=gray!10]  (H11)  at(10,4)   {\color{black} $H_{11}$} ; 
   \node[obs,  rectangle,text=gray, fill=Cerulean!20,very thick,draw=RoyalBlue]  (H12)  at(11,4)   {\color{red}$H_{12}$} ; 
      \node[obs,  rectangle,text=gray, fill=gray!10]  (H13)  at(12,4)   {\color{black} $H_{13}$} ; 
   \node[obs,  rectangle,text=gray, fill=gray!10]  (H14)  at(13,4)   {\color{black} $H_{14}$} ; 
     \node[obs,  rectangle,text=gray, fill=gray!10]  (H15)  at(14,4)   {\color{black} $H_{15}$} ; 
   \node[obs,  rectangle,text=gray, fill=Cerulean!20,very thick,draw=RoyalBlue]  (H16)  at(15,4)   {\color{black} $H_{16}$} ; 

\node[obs, rectangle, text=gray, fill=gray!10]  (H21)  at(.5,3)   {\color{black} $\;\;\;\;\;H^2_{1}\;\;\;\;\;$} ; 
\node[obs, rectangle, text=gray, fill=Cerulean!20]  (H22)  at(2.5,3)   {\color{red}$\;\;\;\;\;H^2_{2}\;\;\;\;\;$} ; 
\node[obs, rectangle, text=gray, fill=Cerulean!20,very thick,draw=RoyalBlue]  (H23)  at(4.5,3)   {\color{red}$\;\;\;\;\;H^2_{3}\;\;\;\;\;$} ; 
\node[obs, rectangle, text=gray, fill=gray!10]  (H24)  at(6.5,3)   {\color{black} $\;\;\;\;\;H^2_{4}\;\;\;\;\;$} ; 
\node[obs, rectangle, text=gray, fill=Cerulean!20,very thick,draw=RoyalBlue]  (H25)  at(8.5,3)   {\color{black} $\;\;\;\;\;H^2_{5}\;\;\;\;\;$} ; 
\node[obs, rectangle, text=gray, fill=Cerulean!20]  (H26)  at(11.5,3)   {\color{red}$\;\;\;\;\;\;\;\;\;\;\;\;\;\;H^2_{6}\;\;\;\;\;\;\;\;\;\;\;\;\;\;$} ; 
\node[obs, rectangle, text=gray, fill=gray!10]  (H27)  at(14.5,3)   {\color{black} $\;\;\;\;\;H^2_{7}\;\;\;\;\;$} ; 

\node[obs, rectangle, text=gray, fill=Cerulean!20]  (H31)  at(1.5,2)   {\color{red}$\;\;\;\;\;\;\;\;\;\;\;\;\;\;H^3_{1}\;\;\;\;\;\;\;\;\;\;\;\;\;\;$} ; 
\node[obs, rectangle, text=gray, fill=Cerulean!20]  (H31)  at(5.5,2)   {\color{red}$\;\;\;\;\;\;\;\;\;\;\;\;\;\;H^3_{2}\;\;\;\;\;\;\;\;\;\;\;\;\;\;$} ; 
\node[obs, rectangle, text=gray, fill=Cerulean!20]  (H31)  at(9.5,2)   {\color{red}$\;\;\;\;\;\;\;\;\;\;\;\;\;\;H^3_{3}\;\;\;\;\;\;\;\;\;\;\;\;\;\;$} ; 
\node[obs, rectangle, text=gray, fill=gray!10]  (H31)  at(13.5,2)   {\color{black} $\;\;\;\;\;\;\;\;\;\;\;\;\;\;H^3_{4}\;\;\;\;\;\;\;\;\;\;\;\;\;\;$} ; 

\node[text width=.25cm] at (0,1)   {$\circ$};
\node[text width=.25cm] at (1,1)  {$\circ$};
\node[text width=.25cm] at (2,1) {\color{red}$\bullet$};
\node[text width=.25cm] at (3,1)  {$\circ$};
\node[text width=.25cm] at (4,1)   {\color{red}$\bullet$};
\node[text width=.25cm] at (5,1)  {$\circ$};
\node[text width=.25cm] at (6,1) {$\circ$};
\node[text width=.25cm] at (7,1)  {$\circ$};

\node[text width=.25cm] at (8,1)   {$\circ$};
\node[text width=.25cm] at (9,1)  {$\circ$};
\node[text width=.25cm] at (10,1) {$\circ$};
\node[text width=.25cm] at (11,1)  {\color{red}$\bullet$};

\node[text width=.25cm] at (12,1)   {$\circ$};
\node[text width=.25cm] at (13,1)  {$\circ$};
\node[text width=.25cm] at (14,1) {$\circ$};
\node[text width=.25cm] at (15,1)  {$\circ$};

   \node[text width=2.5cm] at (17.5,4) 
    {$\mathcal{H}^1$};
    \node[text width=2.5cm] at (17.5,3) 
    {$\mathcal{H}^2$ };
    \node[text width=2.5cm] at (17.5,2) 
    {$\mathcal{H}^3$};
      \node[text width=2.5cm] at (17.5,1) 
    {Signal};

\end{tikzpicture}\end{center}

%% file: 03_elp.tex
We now describe a multiple testing procedure that achieves the goal in (\ref{eq:AdjCount})-(\ref{eq:Disj}). The procedure evaluates the evidence in favor of the hypotheses in $\mathcal H$ by using e-values.
The term e-value has been introduced recently and encompasses multiple related alternative approaches to testing, including betting scores and likelihood ratios \citep{vovk2021valuescalibrationcombination,S21,Get20}. Like p-values, e-values are defined by their properties under the null hypothesis; however, while p-values are defined in terms of probabilities, e-values are defined in terms of expectations. Formally, a p-value $P$ is a random variable that satisfies $\mathbb{P}_{null}(P \leq \alpha) \leq \alpha$ (often with equality) for all $\alpha \in(0,1)$. In other words, a $\mathrm{p}$-variable is stochastically larger than $\mathrm{U}[0,1]$ under the null. An e-value $E$ is a $[0, \infty]$-valued random variable satisfying $\mathbb{E}_{null}[E] \leq 1$.
We can conceptualize it as the stochastic return of a bet against the null hypothesis: on average, if the null is true, it multiplies the money risked by 1 (no gain), and large gains are unlikely (by Markov's inequality a gain larger than a constant $c$ happens with probability at most $1/c$). A large realized e-value, a substantial win in the bet, indicates that the null is unlikely to be true so that the e-value can be interpreted as the amount of evidence collected against the null. Indeed, it has been argued that, for the general public, a summary of evidence given in terms of return on a bet might be easier to interpret than the p-value \citep{S21}.

Aside from this possible increased efficacy in communication, the fact that e-values are defined in terms of expectation translates in a number of advantages from the point of view of designing valid (multiple) testing strategies \citep{wang2022eBH}.
For example, in contrast to p-values, handling dependent e-values is quite straightforward due to the linearity of expectation.
Crucially for our purposes, to establish that a multiple testing procedure controls FDR when it is based on e-values with any dependence structure, it is sufficient to verify that it is self-consistent. This is a purely algorithmic property, originally introduced in \citet{blanchard2008two} with reference to p-values. \citet{wang2022eBH} translate it to the e-value context. Let $e_1, ...., e_{|\mathcal{H}|}$ be e-values associated with hypotheses $H \in \mathcal{H}$. They define an e-testing procedure $\mathcal{F}$ to be self-consistent at level $\alpha \in (0, 1)$ if, for every rejected hypothesis $H_k$, the corresponding e-value $e_k$ satisfies 

\begin{equation}\label{eq:self_consistency_threshold}
 e_k \geq \frac{|\mathcal{H}|}{\alpha R_{\mathcal{F}}},
\end{equation}
where $|\mathcal{H}|$ denotes the total number of hypotheses and $R_{\mathcal{F}}$ denotes the number of rejections. \citet{wang2022eBH} show that self-consistency is a sufficient condition for FDR control, with arbitrary configurations of e-values. This is in contrast to what happens for p-values, where an additional dependency control condition is needed \citep{blanchard2008two}.
Note that the numerator of equation (\ref{eq:self_consistency_threshold}) depends on the number of hypotheses tested, which can make relying on self-consistency for frequentist FDR control conservative.

Similarly to \citet{spector2023controlled}, we can then define a linear program that achieves the goal in (\ref{eq:AdjCount})-(\ref{eq:Disj}). 

\begin{ProcedureDef}{\em \textbf{eLP, e-value linear programming.}}\label{def:lp_eLP}
Let $\mathcal H$ be a multi-resolution family of hypotheses, corresponding to a collection of groups $\mathcal A$, defined starting from $p$ individual hypotheses. Let each hypothesis be associated with an e-value $e_A$ and a weight $w(A)$ for $A \in \mathcal{A}$. Let $\alpha \in (0, 1)$ be the desired level of FDR control and let $x_A \in \{0, 1\}$ be an indicator of whether the hypothesis $H_{A}$ is rejected. The rejection set of eLP is identified by solving the following constrained optimization:
\begin{subequations}
\label{eq:optim}
\begin{align}
 \max_{\{x_{A}\}_{A \in \mathcal{A}}}
 & \quad \sum_{A \in \mathcal{A}} w(A) x_{A} \label{eq:eLP_objective}\\
 \text{subject to} 
 & \quad x_A \in \{0, 1 \} \quad \forall A \in \mathcal{A} \label{eq:eLP_binary}\\
 & \quad |\mathcal{A}| - \alpha e_A \sum_{A \in \mathcal{A}} x_A \leq |\mathcal{A}| \times (1-x_A) \quad \forall A \in \mathcal{A}, \label{eq:eLP_fdr}\\
 & \quad \sum_{A \in \mathcal{A}: j \in A} x_A \leq 1 \quad \forall j \in \{1, ..., p\}. \label{eq:eLP_location_constraint}
\end{align}
\end{subequations}
\end{ProcedureDef}

The first constraint (\ref{eq:eLP_binary}) simply characterizes $x_A \in \{0, 1\}$ as an indicator. The objective (\ref{eq:eLP_objective}) corresponds to our goal in (\ref{eq:AdjCount}).
 The linear constraint (\ref{eq:eLP_fdr}) ensures frequentist FDR control at level $\alpha$. 
The sum $\sum_{A \in \mathcal{A}} x_A$ equals the number of rejections by the procedure. Then, if $x_A = 1$, (\ref{eq:eLP_fdr}) requires that $e_A \geq \frac{|\mathcal{A}|}{\alpha \sum_{A \in \mathcal{A}} x_A}$, which is exactly the self-consistency condition in equation (\ref{eq:self_consistency_threshold}). If $x_A = 0$, (\ref{eq:eLP_fdr}) is automatically satisfied, as $\alpha e_A \sum_{A \in \mathcal{A}} x_A \geq 0$. Constraint (\ref{eq:eLP_location_constraint}) ensures that the rejected regions are disjoint (\ref{eq:Disj}): each $j \in \{1, ..., p\}$ can be rejected at most once, and not multiple times in different groups. 

 \noindent {\bf Remark 1}\quad To use eLP, researchers need to have available e-values for all the hypotheses. In case the testing procedure leads to p-values, it is possible to convert these to e-values using ``calibration,'' albeit with possible power loss (see \citet{vovk2021valuescalibrationcombination}).

 \noindent {\bf Remark 2}\quad A consequence of the linearity of expectation is that an arithmetic mean of e-values of individual hypotheses belonging to a certain group is an e-value for their group-level global null hypothesis. In fact, the arithmetic mean essentially dominates any symmetric e-merging function, i.e.  functions that take e-values of individual hypotheses as input and provide an e-value for the group-level global null as output (see \citep{vovk2021valuescalibrationcombination}, who also discuss cross-merging functions, mapping several p-values into an e-value). Note that if group e-values are constructed with this strategy,  eLP would only reject hypotheses at the finest level of resolution (the individual level). We describe how to get non-trivial group e-values from the knockoff filter in section \ref{sec:kelp}.

\noindent {\bf Remark 3}\quad For a single level of resolution, that is $|\mathcal{M}| = 1$, eLP reduces to the e-BH procedure by \citet{wang2022eBH}.

Self-consistency can also be used to define other multiple comparisons procedures based on e-values. For example, as mentioned in \citet{wang2022eBH}, one can construct the equivalence of Focused BH \citep{katsevich2021filtering} for e-values. We do so precisely in appendix \ref{appendix:focusedeBH_vs_kelp}, obtaining a procedure we call Focused e-BH and which controls FDR for any filter and under any dependence structure. 

Focused e-BH applied to a multi-resolution family of hypotheses with a filter that selects the most precise non-redundant rejections produces the same rejection set of eLP with weights $w(A)=1/|A|$; see appendix \ref{appendix:focusedeBH_vs_kelp} for more details. This is interesting in the context of the comparison in \citet{spector2023controlled} between Focused BH and BLiP, which showed larger power for BLiP. This equivalence relationship between eLP (which has many commonalities with BLiP) and Focused e-BH (which is a translation of Focused BH to e-values) underscores how---when considering exactly the same set of hypotheses---the observed difference in power might not be due to the linear programming formulation, but rather to the implicit estimate of the FDR underlying these two methods. \citet{spector2023controlled} control the Bayesian FDR, which can be evaluated using posterior probabilities. Focused BH and Focused e-BH control the frequentist FDR. A (self-consistent) procedure controlling the frequentist FDR bounds the number of wrong rejections utilizing the total number of hypotheses tested: this can lead to conservative behavior, when non-redundancy constraints limit the number of possible rejections to a value much smaller than the total count of tested hypotheses. To remedy this, \citet{katsevich2021filtering} introduced a permutation based method to estimate the number of false rejections, which, however, comes with computational costs.


%% file: 04_kelp.tex
Given e-values for a family of hypotheses ${\mathcal H}$, eLP guarantees FDR control while leading to resolution-adaptive discoveries. However, defining powerful e-values is a non-trivial task \citep{vovk2021valuescalibrationcombination}. We now describe how to accomplish this for specific types of multi-resolution families of hypotheses leveraging the knockoff framework \citep{barber2015controlling,candes2018panning}.

For a set of $p$ explanatory variables $X = (X_1, ..., X_p)$ and a response $Y$, the knockoff framework tests, with FDR control, conditional independence hypotheses 
\begin{equation}\label{eq:conditional_independence_null}
 H_j: Y \indep X_j \mid X_{-j}
\end{equation}
and, given a partition $\mathcal{A}^m$, their group equivalent
\begin{equation}\label{eq:group_conditional_independence_null}
 H_{A_g^m}: Y \indep X_{A_g^m} \mid X_{-A_g^m}
\end{equation}
where $X_{A_g^m} = \{X_j\}_{j \in A_g^m}$ and $X_{- A_g^m}$ denotes all features except $X_j$ with $j\in A_g^m$ \citep{katsevich2019multilayer,sesia2020multi}. 

The knockoff framework relies on two steps: a) the construction of test statistics with some distributional properties under the null and b) a multiple comparison procedure. The first step a) is based on comparing the signal of each feature (or group of features) with that of a ``knockoff copy'' of it, indistinguishable from the original in distribution, but independent from $Y$ \citep{candes2018panning,sesia2020multi}. This is done with a score $W_j$ designed so that swapping the original feature $X_j$ with its knockoff $\widetilde{X}_j$ flips the sign of $W_j$. For example, a popular choice \citep{candes2018panning} for the scores is the difference of the absolute values of the coefficients based on a cross validated Lasso of $Y$ on the entire collection of original features and knockoffs: $W_j = |\beta_j| - |\widetilde{\beta}_j|$, where $\beta_j$ is the resulting coefficient of $X_j$ and $\widetilde{\beta}_j$ the coefficient of the knockoff $\widetilde{X}_j$. A large value of $W_j$ suggests greater evidence against the null. 

The b) multiple comparison procedure (filter) is based on an estimate of the FDP that relies on the distribution of the knockoff scores $W_j$ corresponding to null hypotheses. Precisely, to control FDR at level $\gamma$, the knockoff filter \citep{barber2015controlling} selects the set of features $\mathcal{R}_{\text{KO}}$ according to the following rule:
\begin{eqnarray}
\mathcal{R}_{\text{KO}} & = & \{j: W_j \geq T\} \label{eq:Rko}\\
 T & = & \inf \{t > 0: \frac{1 + \sum_{j \in [p]} \mathbf{1} \{W_j \leq -t \}}{\sum_{j \in [p]} \mathbf{1} \{W_j \geq t \}} \leq \gamma \} \nonumber.
\end{eqnarray}
The martingale properties that guarantee that (\ref{eq:Rko}) controls FDR at level $\gamma$ also guarantee that:
\begin{equation}\label{equ:martingale_property}
 \mathbb{E} \Big[ \frac{\sum_{j: H_j \in \mathcal{H}_0} \mathbf{1} \{W_j \geq T \}}{1 + \sum_{j: H_j \in \mathcal{H}_0} \mathbf{1} \{W_j \leq -T \}} \Big] \leq 1.
\end{equation}
This expectation motivated the construction in \citet{ren2023derandomized}, who show how the rejections of the knockoff filter are equivalent to those of e-BH applied to e-values defined as
$$e_j:=p \cdot \frac{\mathbf{1}\left\{W_j \geq T\right\}}{1+\sum_{k \in[p]} \mathbf{1}\left\{W_k \leq-T\right\}}.$$
Note that these $e_j$ are not strictly e-values: it is not true that for each null $ \mathbb{E} [e_j]\leq 1$, but only that $ \sum_{j: H_j \in {\cal H}_0} \mathbb{E} [e_j]\leq p$. This property follows from equation (\ref{equ:martingale_property}), as $$\sum_{j: H_j \in {\cal H}_0} \mathbb{E} [e_j] = p \ \ \cdot \sum_{j: H_j \in {\cal H}_0} \mathbb{E} \Big[\frac{\mathbf{1}\left\{W_j \geq T\right\}}{1+\sum_{k \in[p]} \mathbf{1}\left\{W_k \leq-T\right\}}\Big] \leq p \ \ \cdot \mathbb{E} \Big[ \frac{\sum_{j: H_j \in \mathcal{H}_0} \mathbf{1} \{W_j \geq T \}}{1 + \sum_{j: H_j \in \mathcal{H}_0} \mathbf{1} \{W_j \leq -T \}} \Big] \leq p.$$ Nevertheless, \citet{ren2023derandomized} show that this property of ``relaxed'' e-values is sufficient to guarantee FDR control in the rejections defined with e-BH. 

Describing the knockoff filter in terms of e-BH applied to e-values is convenient because of the advantages of this framework in dealing with dependency. \citet{ren2023derandomized} show, for example, how dependent sets of $\{e^k_j, j=1,\ldots p\}$, derived from $k$ runs of the knockoff procedure, can be combined to achieve a stable variable selection.
This is also useful in our multi-resolution framework. To illustrate this, we start generalizing the result by \citet{wang2022eBH} for relaxed e-values.
\begin{theorem}\label{thm:main_thm_elp}
 Suppose $\mathbf{e} = (e_1, ..., e_{n})$ are relaxed e-values for hypotheses $(H_1, \ldots, H_{n})$, i.e. they satisfy $\sum_{j: H_j \in \mathcal{H}_0} \mathbb{E} [e_j] \leq n$. Then any self-consistent e-testing procedure $\mathcal{F}$ at level $\alpha$ taking $\mathbf{e}$ as input controls the FDR at level $\alpha$.
\end{theorem}
The proof is almost identical to the proof in \citet{wang2022eBH} and is included in appendix \ref{appendix:proofs_main}. Theorem \ref{thm:main_thm_elp} guarantees that 
we can use relaxed e-values as an input to eLP (Procedure \ref{def:lp_eLP}). The following describes the construction of knockoffs-based e-values for the hypotheses at each resolution and their use in a procedure we call KeLP. $\widetilde{X}^m$ is said to be a (group) knockoff for partition $m \in \mathcal{M}$ if for each group $A_{g}^m \in \mathcal{A}^m$ it holds that $(X, \widetilde{X}^m)_{\text{swap}({A_{g}^m, \mathcal{A}^m})} \overset{d}{=} (X, \widetilde{X}^m)$ and $Y$ is independent of $\widetilde{X}^m$ given $X$; see e.g. \citet{sesia2021populationstructureshapeit}. $(X, \widetilde{X}^m)_{(\text{swap}_{A_{g}^m, \mathcal{A}^m})}$ means that the $j$-th coordinate is swapped with the $(j + p)$-th coordinate for all $j \in A_g^m$. The construction of (group) knockoff variables has been studied in prior research; see for example \citet{sesia2019gene, sesia2021populationstructureshapeit,dai2016knockoff, katsevich2019multilayer, gimenez2019knockoffs, spector2022powerful, romano2020deep}.

\begin{ProcedureDef}{\em \textbf{(KeLP, knockoff e-value linear programming)}}\label{def:lp_KeLP}
Let ${\cal H}$ be a multi-resolution family of conditional independence hypotheses relative to variables $Y, X_1, \ldots X_p$. For each resolution $m\in {\cal M}$, construct (group) knockoffs for $X_1, \ldots, X_p$ with respect to partition ${\cal A}^m$ and obtain for each hypothesis $H_{A^m_g}$ the knockoff score $W_{A^m_g}$. Let $\{ c_m\}_{m\in{\cal M}}$ be non negative numbers such that $\sum_{m \in \mathcal{M}} c_m \leq |\mathcal{A}|$. 
For each $m \in \mathcal{M}$, define knockoffs-based e-values by
\begin{eqnarray}\label{eq:main_evalue_definition}
 e_{A_{g}^m} & := & c_m \cdot \frac{\mathbf{1} \{W_{A_{g}^m} \geq T^m \}}{1 + \sum_{A \in \mathcal{A}^m} \mathbf{1} \{W_A \leq -T^m \}}, \;\;\; \text{where}\\
\label{eq:stopping_Mtime_gamma}
 T^m & = & \inf \{t > 0: \frac{1 + \sum_{A \in {\cal A}^m} \mathbf{1} \{W_A \leq -t \}}{\sum_{A \in {\cal A}^m} \mathbf{1} \{W_A\geq t \}} \leq \gamma^m \}, \;\; \gamma^m\in (0,1).
\end{eqnarray}
For a fixed weighting function $w(\cdot)$ and a target FDR level $\alpha$, the rejection set of KeLP is given by the output of eLP (Procedure \ref{def:lp_eLP}) taking these and the e-values (\ref{eq:main_evalue_definition}) as input. 
\end{ProcedureDef}

Due to the correspondence between the e-BH procedure by \citet{wang2022eBH} and the traditional knockoff-filter by \citet{candes2018panning}, as described by \citet{ren2023derandomized}, if $|\mathcal{M}| = 1$, KeLP is equivalent to the usual knockoff procedure.

There are two sets of parameters in KeLP: $c_m$ and $\gamma^m$.
In (\ref{eq:main_evalue_definition}), $c_m$ is a fixed multiplier that is chosen to guarantee the property for theorem \ref{thm:main_thm_elp}:
 \begin{equation}\label{eq:knockoff_evalue_fulfill}
 \sum_{A: H_A \in \mathcal{H}_0} \mathbb{E} \Big[ e_A \Big] = \sum_{m \in \mathcal{M}} \sum_{A_g^m: H_{g}^m \in \mathcal{H}_0^{m}} \mathbb{E} \Big[ e_{A_{g}^m} \Big] \leq \sum_{m \in \mathcal{M}} c_m \leq |\mathcal{A}|.
\end{equation} 
The first inequality follows from the property of the importance statistics in equation (\ref{equ:martingale_property}), and the second is true by construction (note that it is possible for some $c_m > |\mathcal{A}^m|$, which relaxes the original definition of knockoff-based e-values by \citet{ren2023derandomized}, since for our goal in (\ref{eq:AdjCount})-(\ref{eq:Disj}) we are not going to apply e-BH to each resolution separately).

The parameter $\gamma^m$ used to obtain the stopping time $T^m$ (\ref{eq:stopping_Mtime_gamma}) governs the trade-off between the number of non-zero e-values and their magnitude. For a larger $\gamma^m$, we might obtain a smaller $T^m$, which might result in a larger number of non-zero e-values. However, those nonzero e-values might be smaller in magnitude, as there might also be more groups with $W_{A_g^m} \leq -T^m$. 

The choice of $c_m$ and $\gamma^m$ is left to the user, and can be guided by tuning in a hold-out dataset. However, note that even in a GWAS setting a ``large'' tuning set might not be required; see appendix \ref{appendix:details_ukb} for details. In this work, we set $c_m = |\mathcal{A}| / |\mathcal{M}|$. This choice results in knockoff e-values whose magnitude depends mostly on signal-strength, and not on the number of groups in a particular resolution. KeLP already includes the preference to reject hypotheses in smaller levels of resolution. As for $\gamma^m$, in applications with moderate dimensions, following \citet{ren2023derandomized}, we recommend $\gamma^m = \alpha / 2$ or $\gamma^m = \alpha / 4$, where $\gamma = \alpha / 4$ should be chosen in higher dimensional settings. We suggest level-specific choices of $\gamma^m$ for very high-dimensional applications with many levels of resolution and very sparse signals (e.g. genetic data): in this case, we recommend choosing $\gamma = \alpha$ for the individual level, with decreasing levels of $\gamma$ as the resolutions becomes coarser. We leave the question of finding the optimal theoretical parameters for further research. We solve the linear program corresponding to KeLP using a standard optimization software: CVXR \citep{cvxr}.

%% file: 05_other_knockoff_eval_applications.tex
Before exploring the performance of KeLP with simulations and real-data analysis, we pause to illustrate other advantages of e-values and of re-casting knockoff analyses in this framework (\ref{eq:main_evalue_definition}). We consider two cases in which the dependency between p-value/test statistics presented challenges for analysis, and underscore the generality of KeLP by showing how it can be leveraged to analyze multivariate structured outcomes.

\subsection{Multilayer Knockoff Filter}\label{subsec:emlkf}

As mentioned in section \ref{sec:notation}, when working with hypotheses at multiple resolutions, different notions of global error are meaningful. The goal of this paper is to reject the most specific hypotheses possible, while controlling the FDR across the rejection set, which might span multiple levels of resolution.

Another goal, referred to as multilayer FDR control, is to coordinate rejections across resolutions and provide simultaneous FDR control in each resolution. 
 \citet{barber2017pfilter} introduced multilayer FDR control and developed the p-filter, which attains it when the p-values for the individual hypotheses are PRDS and group hypotheses are tested using Simes's combination rule \citep{simes1986improved}. 
 
 The original p-filter by \citet{barber2017pfilter} has been extended to arbitrary dependencies between the p-values in \citet{pfilter2019} using reshaping, which makes the procedure conservative.
 
 E-values can be used to develop an analog of the p-filter, as mentioned by \citet{wang2022eBH}. We include a description of the e-filter in appendix \ref{appendix:emlkf_description}.

 \citet{katsevich2019multilayer} developed the Multilayer Knockoff Filter (MKF) to attain the same goal in the context of a multi-resolution analysis via knockoffs. The MKF controls the FDR at level $\alpha \times \kappa$, where $\kappa \approx 1.93$ and $\alpha$ denotes the target FDR level. Using knockoff e-values, we describe the e-Multilayer Knockoff Filter (e-MKF ) with FDR control at the target level $\alpha$ in appendix \ref{appendix:emlkf_description}. We also show that, for the same theoretical level of FDR control, the e-MKF has higher power compared to the MKF.

\subsection{Partial conjunction hypotheses}

A multi-resolution family of hypotheses is one type of ``structured'' collection of hypotheses. Another type that is often useful to consider is the one motivating partial conjunction testing \citep{benjamini2008screeningpartialconjunction}. Here one can think about an array of hypotheses $\{H_j^{\ell}\}_{j=1,\ell=1}^{n,L}$, where for each $j$, researchers are interested to discover that at least $u$ out of $L$ tested hypotheses are false. 
One context where partial conjunctions are quite relevant is replication \citep{BH13,BH18}. For example, the genetic underpinnings of a disorder might be studied in different human populations, and scientists might be interested in identifying those SNPs that carry a signal in at least $u$ distinct groups. 
A related, but different, example, in the context of GWAS, is when measurements on several phenotypes may be available. Scientists might then be interested in testing whether a SNP (or group of SNPs) is conditionally independent of \textit{any} (or \textit{at least} a certain number) of the phenotypes.

\citet{li2021multienvironmentknockoff} developed a knockoff filter to test partial conjunctions across distinct independent studies. This approach, however, cannot be applied to the investigation of multiple phenotypes collected on the same individuals---the knockoff scores relative to the different phenotypes being based on the same genotype data are going to be dependent. 
 Using knockoff e-values we can overcome this limitation, as we describe in the appendix \ref{appendix:global_partial_description}. Testing these partial conjunction hypotheses can also be combined with testing across multiple levels of resolution. Indeed, in our application to the UK Biobank data in section \ref{sec:application_ukb_height}, we test a global null hypothesis for platelet-related outcomes.

\subsection{Structured outcomes}\label{subsec:structured_outcomes_description}

In some cases, it might be possible to describe a hierarchy among the outcomes \citep{cortes2017bayesian}. Let $L$ denote the total number of outcomes. To fix ideas, consider the example in Figure \ref{fig:outcome_tree}. The binary tree has four leaves, two internal nodes and one root node. Each of the parent nodes is constructed as a union of its leaf descendants. 

\begin{figure}[H]
\centering
\input{figures/cartoon_figures/outcome_tree} 
\caption{Illustration of an outcome structure corresponding to a binary tree with seven nodes based on four leaves $A, B, C$ and $D$. Each parent node is constructed based on its children.}
\label{fig:outcome_tree}
\end{figure}
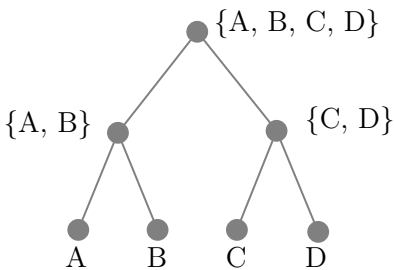

We are interested in testing conditional independence hypotheses between each of the $p$ features and the $L = 7$ different outcomes. Even without considering different groupings of the $p$ features, there is a level of redundancy between the hypotheses. For example, rejecting the hypotheses of conditional independence between feature $X_j$ and both $\{A, B\}$ and $A$ is redundant, as it implicates the leaf-outcome $A$ twice and stating that $X_j$ is not independent from $A$ logically implies that $X_j$ is not independent from $\{A, B\}$.

We can use the KeLP framework to filter rejections so that they point to the most specific outcome for each feature. Specifically, we obtain one knockoff e-value for each feature for each node of the tree shown in Figure \ref{fig:outcome_tree}, i.e. seven knockoff e-values for each feature in total. We require that the hypothesis of conditional independence from each feature be rejected only once for each leaf-outcome. We weigh hypotheses with the reciprocal of the number of leaf-outcomes they implicate. With these specifications, we can then apply KeLP to the $7 \times p$ e-values in total, controlling the FDR across the entire rejected set (across all outcomes and features). 

The output of the procedure returns separately for each feature the most specific outcome for which the conditional independence hypothesis was rejected. We show a simulation of this setting in section \ref{subsec:structured_outcomes_simulation}.

%% file: figures/cartoon_figures/outcome_tree.tex
\tikzset{every picture/.style={line width=0.75pt}} 

\begin{tikzpicture}[x=0.75pt,y=0.75pt,yscale=-1,xscale=1]

\draw  [color={rgb, 255:red, 128; green, 128; blue, 128 }  ,draw opacity=1 ][fill={rgb, 255:red, 128; green, 128; blue, 128 }  ,fill opacity=1 ] (442,407) .. controls (442,404.24) and (444.24,402) .. (447,402) .. controls (449.76,402) and (452,404.24) .. (452,407) .. controls (452,409.76) and (449.76,412) .. (447,412) .. controls (444.24,412) and (442,409.76) .. (442,407) -- cycle ;
\draw  [color={rgb, 255:red, 128; green, 128; blue, 128 }  ,draw opacity=1 ][fill={rgb, 255:red, 128; green, 128; blue, 128 }  ,fill opacity=1 ] (402,458) .. controls (402,455.24) and (404.24,453) .. (407,453) .. controls (409.76,453) and (412,455.24) .. (412,458) .. controls (412,460.76) and (409.76,463) .. (407,463) .. controls (404.24,463) and (402,460.76) .. (402,458) -- cycle ;
\draw  [color={rgb, 255:red, 128; green, 128; blue, 128 }  ,draw opacity=1 ][fill={rgb, 255:red, 128; green, 128; blue, 128 }  ,fill opacity=1 ] (482,457) .. controls (482,454.24) and (484.24,452) .. (487,452) .. controls (489.76,452) and (492,454.24) .. (492,457) .. controls (492,459.76) and (489.76,462) .. (487,462) .. controls (484.24,462) and (482,459.76) .. (482,457) -- cycle ;
\draw  [color={rgb, 255:red, 128; green, 128; blue, 128 }  ,draw opacity=1 ][fill={rgb, 255:red, 128; green, 128; blue, 128 }  ,fill opacity=1 ] (382,507) .. controls (382,504.24) and (384.24,502) .. (387,502) .. controls (389.76,502) and (392,504.24) .. (392,507) .. controls (392,509.76) and (389.76,512) .. (387,512) .. controls (384.24,512) and (382,509.76) .. (382,507) -- cycle ;
\draw  [color={rgb, 255:red, 128; green, 128; blue, 128 }  ,draw opacity=1 ][fill={rgb, 255:red, 128; green, 128; blue, 128 }  ,fill opacity=1 ] (422,507) .. controls (422,504.24) and (424.24,502) .. (427,502) .. controls (429.76,502) and (432,504.24) .. (432,507) .. controls (432,509.76) and (429.76,512) .. (427,512) .. controls (424.24,512) and (422,509.76) .. (422,507) -- cycle ;
\draw  [color={rgb, 255:red, 128; green, 128; blue, 128 }  ,draw opacity=1 ][fill={rgb, 255:red, 128; green, 128; blue, 128 }  ,fill opacity=1 ] (462,507) .. controls (462,504.24) and (464.24,502) .. (467,502) .. controls (469.76,502) and (472,504.24) .. (472,507) .. controls (472,509.76) and (469.76,512) .. (467,512) .. controls (464.24,512) and (462,509.76) .. (462,507) -- cycle ;
\draw  [color={rgb, 255:red, 128; green, 128; blue, 128 }  ,draw opacity=1 ][fill={rgb, 255:red, 128; green, 128; blue, 128 }  ,fill opacity=1 ] (503,508) .. controls (503,505.24) and (505.24,503) .. (508,503) .. controls (510.76,503) and (513,505.24) .. (513,508) .. controls (513,510.76) and (510.76,513) .. (508,513) .. controls (505.24,513) and (503,510.76) .. (503,508) -- cycle ;
\draw [color={rgb, 255:red, 128; green, 128; blue, 128 }  ,draw opacity=1 ]   (407,458) -- (447,407) ;
\draw [color={rgb, 255:red, 128; green, 128; blue, 128 }  ,draw opacity=1 ][fill={rgb, 255:red, 128; green, 128; blue, 128 }  ,fill opacity=1 ]   (487,457) -- (447,407) ;
\draw [color={rgb, 255:red, 128; green, 128; blue, 128 }  ,draw opacity=1 ]   (508,508) -- (496.69,480.54) -- (487,457) ;
\draw [color={rgb, 255:red, 128; green, 128; blue, 128 }  ,draw opacity=1 ]   (407,458) -- (387,507) ;
\draw [color={rgb, 255:red, 128; green, 128; blue, 128 }  ,draw opacity=1 ]   (407,458) -- (427,507) ;
\draw [color={rgb, 255:red, 128; green, 128; blue, 128 }  ,draw opacity=1 ]   (487,457) -- (467,507) ;

\draw (379,514) node [anchor=north west][inner sep=0.75pt]   [align=left] {A};
\draw (420,514) node [anchor=north west][inner sep=0.75pt]   [align=left] {B};
\draw (460,514) node [anchor=north west][inner sep=0.75pt]   [align=left] {C};
\draw (500,514) node [anchor=north west][inner sep=0.75pt]   [align=left] {D};
\draw (348,445) node [anchor=north west][inner sep=0.75pt]   [align=left] {\{A, B\}};
\draw (500,443) node [anchor=north west][inner sep=0.75pt]   [align=left] {\{C, D\}};
\draw (454,394) node [anchor=north west][inner sep=0.75pt]   [align=left] {\{A, B, C, D\}};

\end{tikzpicture}

%% file: 06i_simulation.tex
To illustrate KeLP's power in localizing signals among features we use two simulation settings: in the first, the features are generated with a block-diagonal variance covariance matrix and in the second they are genotypes of SNPs on chromosome 21 of the White Non-British population of the UK Biobank. 
We also explore KeLP's performance with structured outcomes, as in Figure \ref{fig:outcome_tree}. The code for the simulations is available at: \url{https://github.com/pmgblz/KeLP}.

\subsection{Block-diagonal correlation structure for features}\label{subsec:blockdiagonal}

We simulate $n$ observations of the outcome from $Y \mid X \sim N(X^T \beta, I_n)$, with $I_n$ the identity matrix. The feature vector $X \in \mathbb{R}^p$ is generated as $X \sim N(0, \Sigma)$, where $\Sigma$ is block-diagonal with blocks of size 5 and $p = 1,000$. Within each block, the correlation pattern follows an AR(1) process with $\Sigma_{jk} = 0.8^{|j-k|}$. This  provides a schematic representation of the dependence structure between DNA polymorphisms. Following further \citet{spector2023controlled}, there are $s \times p$ randomly chosen nonzero coefficients, rounded to the nearest integer, where $s \in (0, 1)$ denotes the sparsity. The nonzero coefficients are simulated as i.i.d. $N(0, \tau^2)$ with $\tau = 0.2$. Across the simulations, we vary the ratio $n/p$, starting at $n/p = 0.1$. 

We consider two levels of resolutions: the individual level (groups of size 1) and groups of size 5. The (group) knockoff variables are generated following the maximum entropy criteria \citep{spector2022powerful, knockoffjulia}; see appendix \ref{appendix:details_simulation} for details. 

We compare KeLP's rejection sets with those of the knockoff filter applied separately at each level of resolution (``Knockoffs individual'' and ``Knockoffs group''), as well as two other approaches to identify rejections across resolutions. The first (``Knockoffs outer'') consists in simply identifying the outer-nodes of the rejections that the knockoff filter makes at each resolution. Outer-nodes are those groups of variables with no rejected descendants, where a descendant is a group contained within the parent group (see Figure \ref{fig:KnockoffZoom}). Note that filtering the resolution-specific rejections by the knockoff filter to the outer-nodes does not have guarantees on FDR control. To obtain FDR control at level $\alpha$ on the filtered outer-nodes, it is possible to run e-BH on their corresponding knockoff e-values at the amended level $\widetilde{\alpha} = \alpha |S| / |\mathcal{H}|$, where $|S|$ denotes the number of filtered outer-nodes. This ensures self-consistency; see \citet{wang2022eBH}. We include this procedure (``e-BH knockoffs outer'') in our comparisons.

To evaluate performance, we look at FDP, power and precision of discoveries.
The realized FDP is calculated separately for each level of resolution for the knockoff filter (individual and group) and across the resolutions for the other methods.
We define ``power'' as the number of correctly rejected nonzero individual features (whether they were rejected individually or in a group), divided by the total number of nonzero individual features. 
To evaluate the precision of discoveries we look at the total number of features included in the rejected groups. For a given level of power, the most precise method points to the smallest number of features.

Figure \ref{fig:main_simulation} shows that the methods perform as expected in terms of FDR control. In terms of power and size of the rejection set, KeLP effectively ``interpolates'' between the individual and the group level. For lower levels of signal strength, KeLP rejects more groups than individual features, but achieves the precision of the individual level as signal strength increases. KeLP's power is very close to that of the outer-nodes of the knockoff filter applied separately at each resolution, but achieves FDR control. Refining the outer-nodes with e-BH leads to no power.

\begin{figure}[H]
\centering
\begin{tabular}{c}
 \includegraphics[scale = 0.4]{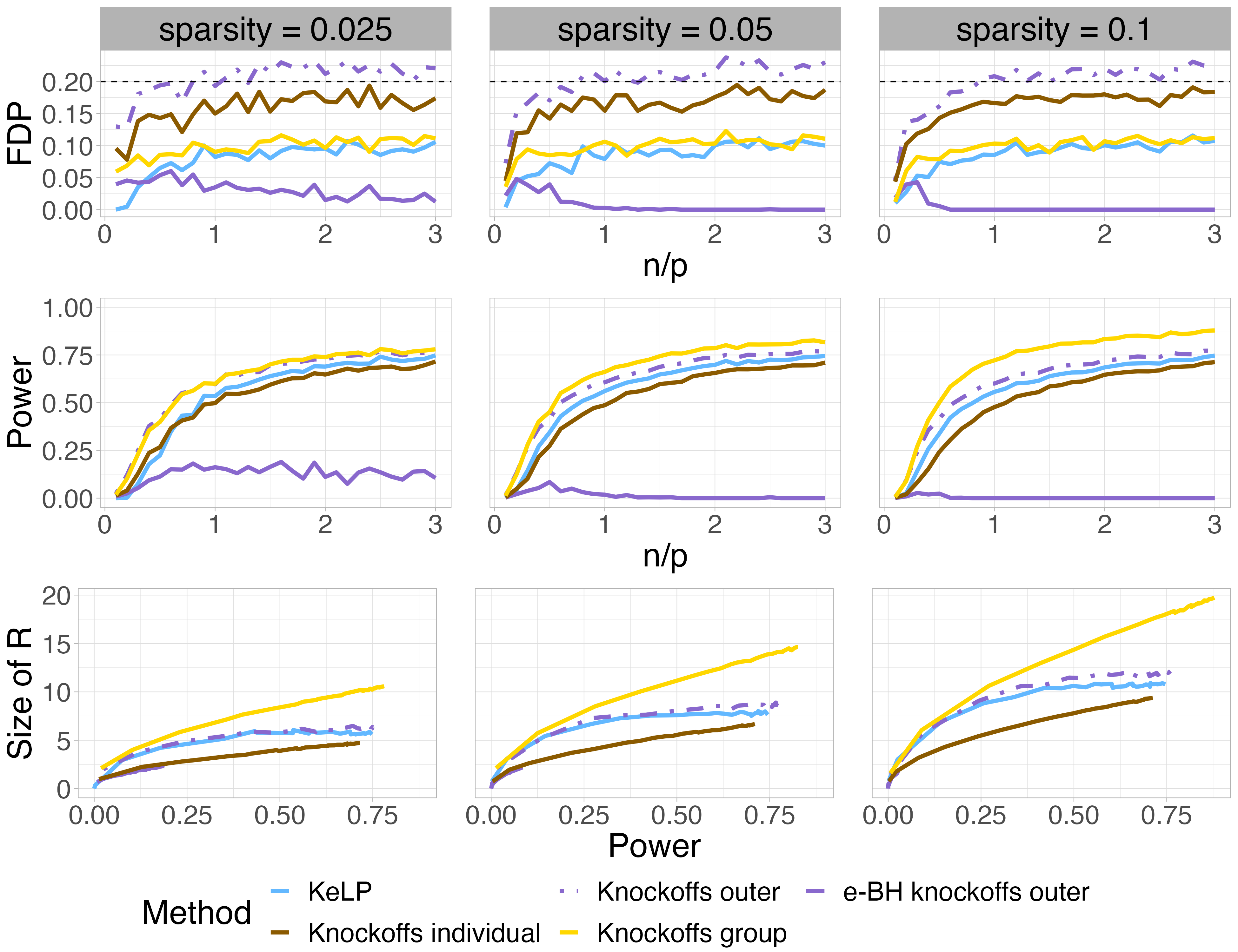} 
\end{tabular}
\caption{By row, FDP, power, and size of the rejection set, averaged across 100 simulation runs. By columns, different signal sparsities (corresponding signal to noise ratios are 1, 2 and 4, from left to right). Target FDR is indicated with a dashed line. The size of the rejection set is displayed as the square root of the total number of variables included in the rejected hypotheses. Broken lines are used for methodologies not expected to control FDR. Colors in the blue-purple spectrum indicate procedures whose rejection sets span across multiple resolutions. }
\label{fig:main_simulation}
\end{figure}

\subsection{Chromosome-wide simulation on the UK Biobank genotypes}\label{subsec:chromosome_simulation_UKB}

We again simulate $Y \mid X \sim N(X^T \beta, I_n)$. Here, $X$ are the genotypes of chromosome 21 of the unrelated White Non-British population of the UK Biobank with $p = 8,832$ and $n = 14,733$. We use the ancestries, resolutions and knockoffs defined by \citet{sesia2021populationstructureshapeit}; see also appendix \ref{appendix:details_ukb} for more details. We compare the same methods as in section \ref{subsec:blockdiagonal}. 
There are $s \times p$ randomly chosen nonzero coefficients across chromosome 21, rounded to the nearest integer. As $n$ and $p$ are fixed in this setting, to explore different signal strengths, we vary the absolute value of the nonzero elements $\beta$: these are set to $a / \sqrt{n}$, where $a$ denotes the signal amplitude. The signs of the nonzero coefficients are determined by independent coin flips.

\begin{figure}[H]
\centering
\begin{tabular}{c}
 \includegraphics[scale = 0.3]{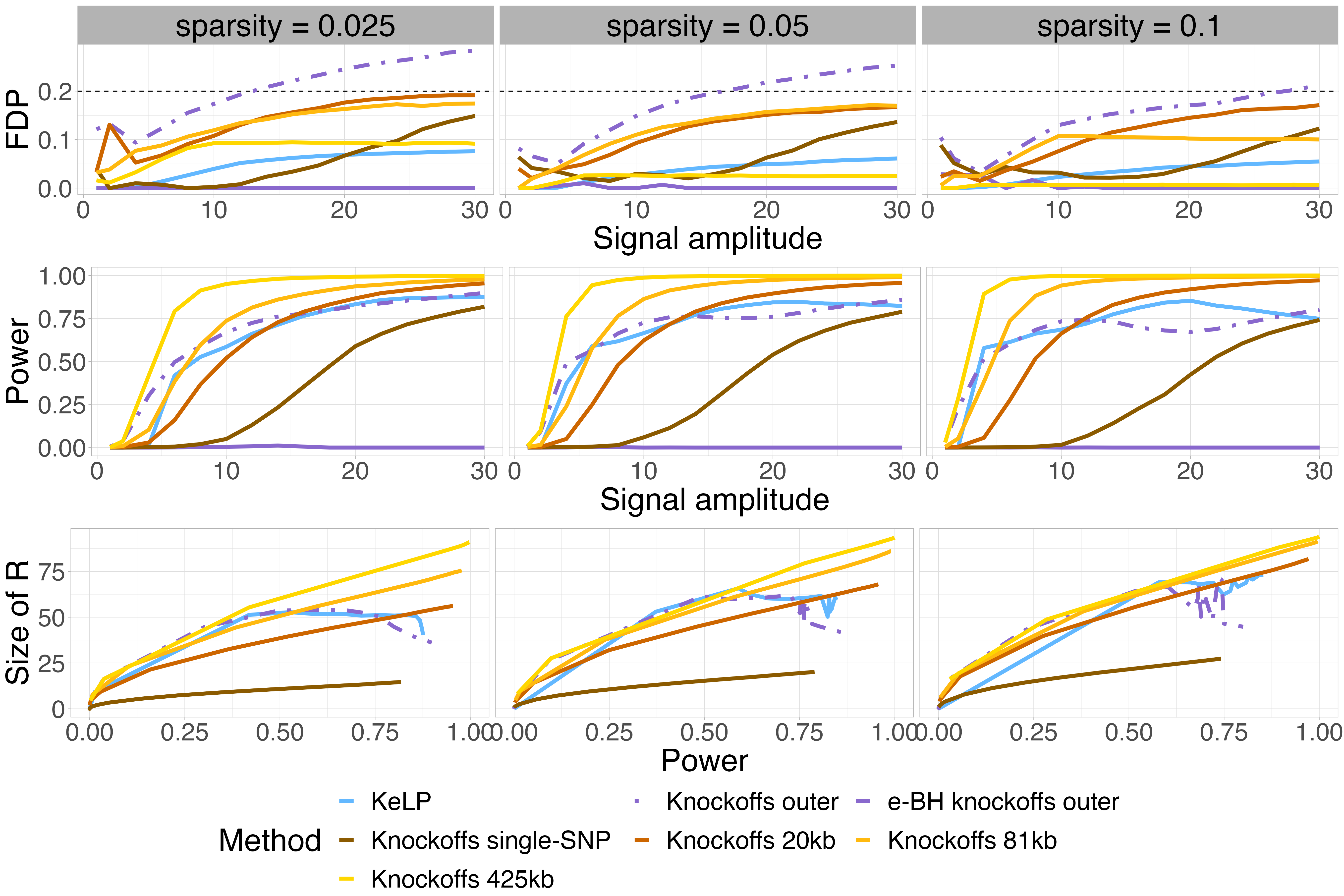} 
\end{tabular}
\caption{By row, FDP, power, and size of the rejection set, averaged across 25 simulation runs. By columns, different signal sparsities. Figure \ref{fig:UKB_simulation_snr} in appendix \ref{appendix:details_simulation} shows heritability ($\text{Var}(f(X)) / \text{Var}(Y) = \text{Var}(X \beta) / \text{Var}(X \beta + \epsilon)$) as a function of signal amplitude. See caption of Figure \ref{fig:main_simulation} for details. Resolutions are indicated by median group width in kilobases (kb); see \citet{sesia2021populationstructureshapeit} for further details.}
\label{fig:UKB_simulation}
\end{figure}

Figure \ref{fig:UKB_simulation} shows similar results to those in Figure \ref{fig:main_simulation}: KeLP leads to a rejection set comparable to that of knockoffs outer-nodes, while controlling the FDR at the desired level.

\subsection{Structured outcomes}\label{subsec:structured_outcomes_simulation}

We consider outcomes with a structure that corresponds to the tree in Figure \ref{fig:outcome_tree}. In particular, consider a setting where the outcomes corresponding to leaves refer to the most specific diseases. We denote by cases the proportion of people who have these diseases. Let $\ell \in \{1, .., L\}$ with $L = 7$ indicate the node. Each of the four leaves is simulated from a logistic model $Y_\ell \mid X \sim \text{Bernoulli} ((1 + e^{-(\delta_\ell + X^T \beta_{\ell})})^{-1})$ where $\delta_\ell$ is an intercept chosen to obtain an approximate percentage of cases of 15\% for each of the leaves. We simulate $X \in \mathbb{R}^p$ as $X \sim N(0, \Sigma)$, where $\Sigma$ has entries $\Sigma_{jk} = 0.3^{|j-k|}$ and $p = 1,000$. There are $s \times p$ randomly chosen nonzero coefficients, rounded to the nearest integer, where $s \in (0, 1)$ denotes the sparsity. We also set an overlap of 50\% for the nonzero features between each of the siblings corresponding to the same internal parent node. The magnitudes of the nonzero coefficients are simulated as i.i.d. $N(\mu, \tau^2)$ with $\mu = 1$ and $\tau = 0.5$. We define the binary outcomes corresponding to the parent nodes in the tree to be equal to $1$ whenever any of their descendants is equal to $1$. Across the simulations, we vary the ratio $n/p$, starting at $n/p = 0.1$. 

In this hierarchical setting, it makes sense to consider the properties of rejection sets at each level of the tree, or among outer-nodes. We compare the performance of KeLP, the combined results of the knockoff filter for the outcomes at each level of the tree (``knockoffs leaves'', ``knockoffs internal nodes'', ``knockoffs root node'') and the corresponding outer-nodes. 

The knockoff filter for the outcome corresponding to the root node has FDR control. However, the rejection set obtained by combining the results of the knockoff filter run separately for the two outcomes corresponding to internal nodes does not have FDR control (as we are simultaneously considering multiple outcomes). Neither does the rejection set that collects all discoveries made by the knockoff filter separately run on the leaf-outcomes. As before, to obtain FDR control, we pass these through a further e-BH procedure at the amended level $\widetilde{\alpha} = \alpha |S| / |\mathcal{H}|$ (``e-BH knockoffs leaves'', ``e-BH knockoffs internal''), where $|S|$ denotes the number of rejections by the knockoff filter collected across multiple outcomes, but within the same level of the tree, and $|\mathcal{H}|$ the total number of hypotheses at the specific level of the tree.

Consider the tree in Figure \ref{fig:outcome_tree}: each leaf-outcome is associated with a set of nonzero features. We call the sum of the number of nonzero features across the leaf-outcomes the ``total number of nonzero leaf-outcome associations''. Moreover, if, for example, feature $X_j$ is nonzero for outcome $A$, it is also nonzero for the outcome corresponding to the internal node $\{A, B\}$ and the root node $\{A, B, C, D\}$. Therefore, if $X_j$ is rejected at the level of the internal node or the root node, it is also counted as a correct rejection. We then define power as the total number of correct rejections divided by the total number of nonzero feature-outcome associations. We calculate power and FDR separately at each level of the tree for the rejection set obtained by combining the results of the knockoff filter at each level and across levels for KeLP and the knockoffs outer-nodes. 

The discovery of an association between $X_j$ and both $\{A\}$ and $\{B\}$ is more precise than the discovery of an association between $X_j$ and $\{ A,B\}$. In agreement with the weighting scheme we described in section \ref{subsec:structured_outcomes_description}, we measure precision in terms of number of ``individually'' associated outcomes. The rejection of the hypothesis of conditional independence between $X_j$ and each of the leaves $A,B,C,D$ adds 1 to this count; the rejection of an internal node adds 1/2, and the rejection of the root node adds 1/4.

\begin{figure}[H]
\centering
\begin{tabular}{c}
 \includegraphics[scale = 0.4]{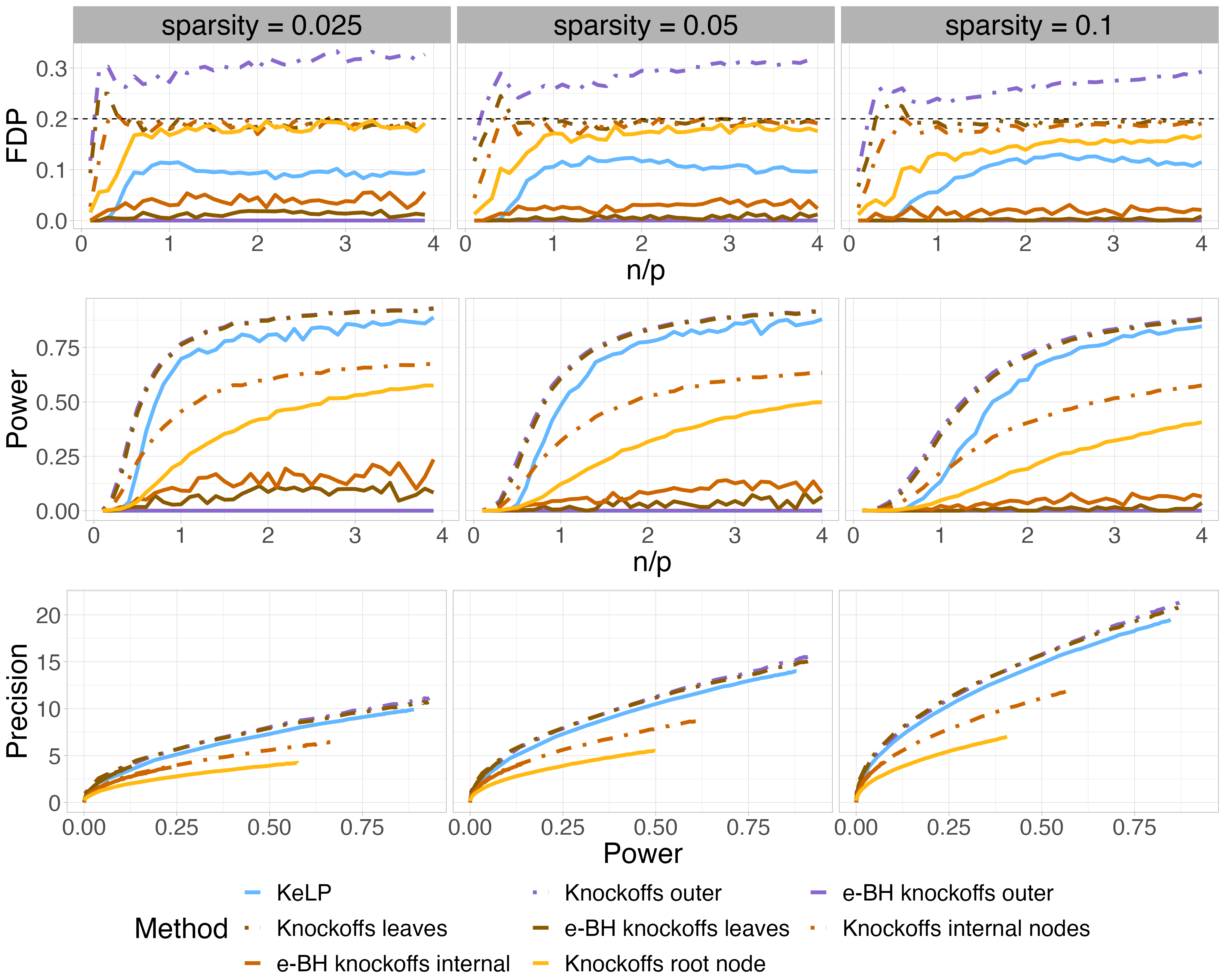} 
\end{tabular}
\caption{By row, FDP, power, and the square root of the number of singularly implicated outcomes (``precision''), averaged across 100 simulation runs. By columns, different signal sparsities. Target FDR is indicated with a dashed line. Broken lines are used for methodologies not expected to control FDR. Colors in the blue-purple spectrum indicate procedures whose rejection sets span across multiple levels of the tree. Note that, unlike in Figures \ref{fig:main_simulation} and \ref{fig:UKB_simulation}, higher values are preferred in the bottom panel. The power and precision curves for ``Knockoffs outer'' and ``Knockoffs leaves'' almost coincide.}
\label{fig:hierarchical_simulation}
\end{figure}

Figure \ref{fig:hierarchical_simulation} shows that KeLP has FDR control, as expected from the theoretical results in section \ref{sec:kelp}. Moreover, in terms of power and precision, it behaves similarly to the knockoffs outer-nodes. However, the knockoffs outer-nodes violate FDR control. Moreover, KeLP has highest power among all methods with theoretical FDR control (all solid lines in Figure \ref{fig:hierarchical_simulation}).

%% file: 06ii_UKB_application.tex
We use KeLP to analyze the relation between two phenotypes (height and platelets) and genetic variation in approximately 337k British unrelated individuals of the UK Biobank. The genotype data consists of about 592k autosomal SNPs. We consider seven different partitions of SNPs corresponding to different resolutions, as in \citet{sesia2021populationstructureshapeit} and rely on the knockoffs generated in that study to construct e-values (\ref{eq:main_evalue_definition}) for all group hypotheses. Our target FDR level is $\alpha = 0.1$.

Table \ref{tab:UKB_rejections_height} shows, stratified by resolution, the number of rejections and of implicated SNPs by (1) KeLP, (2) the knockoff filter applied to each resolution separately, and (3) the outer-nodes of the knockoffs rejections. Appendix \ref{appendix:height_portion_genome_implicated} provides a visualization of the KeLP rejection regions.

\begin{center}

\begin{ThreePartTable}
                  
	              \begin{TableNotes}
									\item British unrelated population. FDR level: $\alpha = 0.1$. The column ``Res. spec.'' (resolution specific) reports the results of the knockoff filter applied separately at each level of resolution (these rejection sets have FDR control at the resolution level, but are redundant across resolutions). The column ``Outer-nodes'' refers to the results of filtering these resolution-specific knockoffs rejections (which leads to a rejection set with no FDR control). Resolutions are indicated by median group width in kilobases (kb); see \citet{sesia2021populationstructureshapeit} for further details.

								\end{TableNotes}

								  \scriptsize
								\begin{longtable}{@{\extracolsep{2pt}} lllllllllll} 
								\caption{UK Biobank: height} 	
									\input{tables/UKB/actual_rejections/height_british_unrelated_UKB_comparison}\label{tab:UKB_rejections_height}

								\end{longtable}

\end{ThreePartTable}

\end{center}

To explore the validity of the KeLP rejections, we compare them to the discoveries in \citet{yengo2022saturated}, who find 12,111 independent SNPs that are significantly associated with height on the basis of a COJO (conditional and joint association analysis with summary statistics) analysis based on GWAS data of 5.4 million individuals. For this analysis, we consider groups which do not contain any of these SNPs (matching by name or location) as true nulls, and evaluate their discoveries as false. The derived estimate of FDP for KeLP is around 7.5\%, while that of the knockoffs outer-nodes is approximately 20\%---substantially higher than the target FDR level of 10\%. The rejections by KeLP include 4,121 SNPs out of the 12,111 SNPs in \citet{yengo2022saturated}, while 5,203 of these are discovered by the knockoffs outer-nodes: the increased FDP of this last approach comes with low power rewards.

We consider four continuous platelet-related phenotypes (platelet count, platelet crit, platelet width, platelet volume) and look for locations in the genome associated to any of them. To obtain global null e-values for each group at each level of resolution, we average the knockoff e-values across the outcomes; see appendix \ref{appendix:global_partial_description} for details. 
Table \ref{tab:UKB_rejections_global} summarizes the results.

\begin{center}

\begin{ThreePartTable}
                  
	              \begin{TableNotes}
									\item  See caption of Table \ref{tab:UKB_rejections_height} for detailed explanations.

								\end{TableNotes}

								  \scriptsize
								\begin{longtable}{@{\extracolsep{2pt}} llllllllll} 
								\caption{UK Biobank: platelet global null}	
									\input{tables/UKB/actual_rejections/platelet_platelet_volume_platelet_width_platelet_crit_british_unrelated_UKB_comparison}\label{tab:UKB_rejections_global}

								\end{longtable}

\end{ThreePartTable}

\end{center}

It is interesting to remark that KeLP can increase power for higher levels of resolutions compared to the standard ``resolution specific'' analysis. Indeed, the rejections by KeLP are not necessarily a subset of the outer-nodes of these. KeLP considers all hypotheses in the multi-resolution family jointly, which can result in a different threshold required for self-consistency than a resolution-specific threshold; see equation (\ref{eq:self_consistency_threshold}). The fact that a number of low resolution hypotheses are discovered, can then lower the threshold for high resolution ones to enter the rejection set.

%% file: tables/UKB/actual_rejections/height_british_unrelated_UKB_comparison.tex
\\[-1.8ex]\hline 
\hline \\[-1.8ex] 
& & \multicolumn{3}{c}{Number of rejections} & & \multicolumn{3}{c}{SNPs implicated}  \\ 
\cline{3-5} \cline{7-9} \\
Resolution &  & KeLP & Outer-nodes & Res.  spec.  &  & KeLP & Outer-nodes & Res.  spec.  \\ 
\hline \\[-1.8ex] 
single-SNP & $$ & $53$ & $53$ & $53$ & $$ & $53$ & $53$ & $53$ \\ 
3 kb & $$ & $148$ & $260$ & $313$ & $$ & $1,545$ & $2,549$ & $3,010$ \\ 
20 kb & $$ & $245$ & $932$ & $1,215$ & $$ & $6,587$ & $20,455$ & $26,909$ \\ 
41 kb & $$ & $345$ & $1,024$ & $2,139$ & $$ & $12,079$ & $31,480$ & $68,070$ \\ 
81 kb & $$ & $711$ & $891$ & $2,817$ & $$ & $34,715$ & $38,732$ & $132,493$ \\ 
208 kb & $$ & $555$ & $882$ & $3,159$ & $$ & $47,387$ & $65,595$ & $267,712$ \\ 
425 kb & $$ & $406$ & $371$ & $2,783$ & $$ & $55,338$ & $44,431$ & $396,238$ \\ 
\hline \\[-1.8ex] 
Total & $$ & $2,463$ & $4,413$ & $$ & $$ & $157,704$ & $203,295$ & $$ \\ 
\hline \\[-1.8ex] 
\insertTableNotes

%% file: tables/UKB/actual_rejections/platelet_platelet_volume_platelet_width_platelet_crit_british_unrelated_UKB_comparison.tex
\\[-1.8ex]\hline 
\hline \\[-1.8ex] 
& & \multicolumn{3}{c}{Number of rejections} & & \multicolumn{3}{c}{SNPs implicated}  \\ 
\cline{3-5} \cline{7-9} \\
Resolution &  & KeLP & Outer-nodes & Res.  spec.  &  & KeLP & Outer-nodes & Res.  spec.  \\ 
\hline \\[-1.8ex] 
single-SNP &   & 116 & 0 & 0 &   & 116 & 0 & 0 \\ 
3 kb &   & 77 & 96 & 96 &   & 783 & 954 & 954 \\ 
20 kb &   & 286 & 354 & 438 &   & 7,162 & 8,446 & 10,531 \\ 
41 kb &   & 143 & 146 & 539 &   & 5,003 & 5,123 & 19,010 \\ 
81 kb &   & 528 & 534 & 1,071 &   & 28,597 & 28,959 & 57,813 \\ 
208 kb &   & 334 & 334 & 1,268 &   & 30,815 & 30,815 & 123,340 \\ 
425 kb &   & 207 & 207 & 1,275 &   & 28,816 & 28,816 & 204,169 \\ 
\hline \\[-1.8ex] 
Total &   & 1,691 & 1,671 &   &   & 101,292 & 103,113 &   \\ 
\hline \\[-1.8ex] 
\insertTableNotes

%% file: 07_conclusion.tex
We developed a method to localize signals at the smallest level of resolution possible while controlling frequentist FDR leveraging the concept of e-values. In particular, we focused on relaxed multi-resolution e-values from the knockoff filter. We have shown that our method can be used to parse over groups of features or groups of outcomes. Our method showed desirable performance in simulations and in an application to the UK Biobank.
It successfully navigated the tradeoff between resolution and power, adaptively choosing a level of resolution that enables discoveries while maximizing their precision and controlling FDR.

The rejections of KeLP are not necessarily a subset of the outer-nodes of resolution-specific knockoffs filters. Indeed, by looking simultaneously across resolutions, and allowing discoveries of different levels of precision, KeLP can lead to a larger number of high resolution discoveries.

We observed in our experiments that KeLP achieves an empirical FDR that is usually lower than the target level. 
This might be connected to the self-consistency requirement, which, as stated in section \ref{sec:elp}, implicitly estimates the number of false discoveries as a fraction of all tested hypotheses, even though our non-redundancy requirement could never result in these many rejections. The investigation of sharper bounds that would result in higher power might be the object of future work. Another possible direction for future research is to investigate whether logical constraints could be utilized for higher power, similar to \citet{meijer2015multiple} or \citet{ramdas2017dagger}.

Our results include methods to test conditional independence partial conjunction hypotheses using the same sample and a new version of the multilayer knockoff filter. All together, they underscore the versatility of e-values and provide examples where the technical elegance of this methodology is accompanied by substantial power.

%% file: 08i_appendix_proof_main.tex
\begin{proof}

The proof is essentially identical to the proof of proposition 2 in \citet{wang2022eBH}, but is included for completeness.

Let $\mathcal{F}$ be a self-consistent e-testing procedure and $\mathbf{e} = (e_1, ..., e_{n})$ be an arbitrary vector relaxed e-variables satisfying $\sum_{j:H_j \in \mathcal{H}_0} \mathbb{E} [e_j] \leq n$. 

This first part is identical to the proof of proposition 2 in \citet{wang2022eBH}:

\begin{align*}
    \text{FDP}(\mathcal{F}) &= \frac{|\mathcal{F}(\mathbf{e}) \cap \mathcal{H}_0|}{R_{\mathcal{F}} \vee 1} \\
    &= \sum_{j:H_j \in \mathcal{H}_0} \frac{\mathbf{1} \{j \in \mathcal{F}(\mathbf{e}) \}}{R_{\mathcal{F}} \vee 1} \\
    &\leq \sum_{j:H_j \in \mathcal{H}_0} \frac{\mathbf{1} \{j \in \mathcal{F}(\mathbf{e}) \} \alpha e_j}{n} \\
    &\leq \sum_{j:H_j \in \mathcal{H}_0} \frac{\alpha e_j}{n} 
\end{align*}

Further following \citet{wang2022eBH}, but using knockoff e-values we get: 

\begin{align*}
    \text{FDR}(\mathcal{F}) &= \mathbb{E} \Big[ \frac{|\mathcal{F}(\mathbf{e}) \cap \mathcal{H}_0|}{R_{\mathcal{F}} \vee 1} \Big] \\
    &\leq \mathbb{E} \Big[ \sum_{j:H_j \in \mathcal{H}_0} \frac{\alpha e_j}{n} \Big] \\
    &=  \frac{\alpha}{n} \sum_{j:H_j \in \mathcal{H}_0} \mathbb{E} \Big[  e_j \Big] \\
    &\leq \alpha
\end{align*}

The only change to the proof in \citet{wang2022eBH} is the last inequality. Here, the last inequality is because $\sum_{j:H_j \in \mathcal{H}_0} \mathbb{E} \Big[  e_j \Big] \leq n$ by assumption.

\end{proof}

%% file: 08ii_appendix_focusedeBH_vs_elp.tex
Following \citet{katsevich2021filtering}, let $\mathcal{R} \subseteq \mathcal{H}$ denote a collection of rejections from a multiple comparison procedure. Further, let $\mathfrak{F}$ be a filter, a function that selects some elements of the rejection set (according to their level of interest, for example, or to ensure non redundancy): $\mathfrak{F}: (\mathcal{R}, \mathbf{e}) \mapsto \mathcal{U} \subseteq \mathcal{H}$ with $\mathcal{U} \subseteq \mathcal{R}$.

We then define:

\begin{algorithm}[H]
\label{alg:focusedeBH}
\caption{\textbf{Focused e-BH}}
\KwIn{Vector of e-values $\mathbf{e} = (e_1, ..., e_{|\mathcal{A}|}$), filter $\mathfrak{F}$.}
\KwOut{Filtered rejection set $\mathcal{U}^*$}
\For{$t \in \{e_1, ..., e_{|\mathcal{A}|}, \infty \}$}{
  Compute $t \times |\mathfrak{F}(\{A \in \mathcal{A}: e_A \geq t \}, \mathbf{e}) | \vee 1)$  
}
$t^* \gets \inf \Bigl\{ t \in \{e_1, ..., e_{|\mathcal{A}|}, \infty\}: t \times |\mathfrak{F}(\{A \in \mathcal{A}: e_A \geq t \}, \mathbf{e}) |\geq |\mathcal{A}| / \alpha \Bigl\}$ \;
 Find the base rejection set $\mathcal{R}^* = \{A: e_A \geq t^* \}$\;
 Compute $\mathcal{U}^* = \mathfrak{F}(\mathcal{R}^*, \mathbf{e})$
\end{algorithm}

This algorithm fulfills self-consistency by definition and hence controls the FDR at level $\alpha$ without making any assumptions on the dependence among test statistics, structure of hypotheses or form of filter, as alluded to in \citet{wang2022eBH}. In contrast, the analogous Focused BH procedure by \citet{katsevich2021filtering} controls FDR when the filter fulfills some monotonicity requirements and/or the p-values are PRDS. While \citet{katsevich2021filtering} also use reshaping \citep{benjamini2001control} to define a version of the procedure that is valid more generally, the resulting rejection set is conservative.

It is interesting to note that there is a close connection between eLP and Focused e-BH. For example, consider a multiresolution family of hypotheses defined by nested partitions, so that logical relations between hypotheses can be described by a tree. Further, consider two procedures: eLP with weights $w(A)$ decreasing in group size $|A|$ and Focused e-BH with the outer-nodes filter, which selects the rejections in $\cal R$ that are not parent to any other rejection. Let ${\cal R}_{on}$ be the output of Focused e-BH with the outer-nodes filter: ${\cal R}_{on}$ is non-redundant by definition; moreover, by focusing on the outer-nodes, ${\cal R}_{on}$ also selects the rejections with highest weight $w(A)$ among redundant sets, which is the goal of eLP. The two procedures, then, lead to the same set of rejections.

The outer-nodes filter is one example of a filtering procedure that can be described as a weighted maximization problem subject to constraints. Focused e-BH in combination with these filters results in the same set of rejections as eLP (where the constraints and weights for eLP could possibly be modified from Procedure \ref{def:lp_eLP}). Below, we focus on filters with a non-redundancy constraint, as this is one of the main motivating observations in \citet{katsevich2021filtering}.

Consider the for-loop in the Focused e-BH procedure (Algorithm S.\ref{alg:focusedeBH}) and the definition of eLP (Procedure \ref{def:lp_eLP}). We can write the inside of the for-loop as a maximization problem. For any $t \in \{e_1, ..., e_{|\mathcal{A}|}, \infty\}$, applying a filter that can be described as a weighted maximization problem with a non-redundancy constraint is equivalent to 

\color{black}
 
\begin{subequations}
\label{eq:optim_loop_kelp_febh}
\begin{align}
  \max_{\{x_{A}\}_{A \in \mathcal{A}}}
    & \quad \sum_{A \in \mathcal{A}} w(A) x_{A} \label{eq:kelp_objective_febh_loop}\\
  \text{subject to} 
  & \quad x_A \in \{0, 1 \} \quad \forall A \in \mathcal{A} \label{eq:kelp_binary_febh_loop}\\
    & \quad e_A \times x_A \geq t \times x_A \quad \forall A \in \mathcal{A}, \label{eq:kelp_febh_rejection_threshold}\\
    & \quad \sum_{A \in \mathcal{A}: j \in A} x_A \leq 1 \quad \forall j \in \{1, ..., p\}, \label{eq:kelp_location_constraint_febh_loop}.
\end{align}
\end{subequations}

Here, constraint (\ref{eq:kelp_febh_rejection_threshold}) does not control FDR, but ensures to only reject hypotheses with e-values $e_A \geq t$.
Under the constraint that every rejected e-value fulfills $e_A \geq t$ and that the rejections are non-redundant (\ref{eq:kelp_location_constraint_febh_loop}), the objective in (\ref{eq:kelp_objective_febh_loop}) maximizes the sum of weights of the rejected groups. For example, consider the outer-nodes filter: among all sets of groups $A \in \mathcal{A}$ with $e_A \geq t$ that would fulfill non-redundancy (\ref{eq:kelp_location_constraint_febh_loop}), the outer-nodes maximize $\sum_{A \in \mathcal{A}} w(A) x_A$ by definition. 

Next, consider choosing $t^*$ (e.g. with the outer-nodes filter for nested groups of hypotheses) in the Focused e-BH procedure. By definition, this chooses the largest number of rejections, controlling the FDR; see algorithm S.\ref{alg:focusedeBH}. Suppose that eLP and Focused e-BH were to choose two different thresholds, say $t^* = t_{\text{F-eBH}} < t_{\text{eLP}}$, which both control the FDR and fulfill the non-redundancy constraint. The objective of eLP (\ref{eq:eLP_objective}) is similarly to maximize the number of weighted rejections. Thus, we can achieve at least the same objective (\ref{eq:eLP_objective}) with $t_{\text{F-eBH}}$ compared to $t_{\text{eLP}}$. Therefore, choosing $t^*$ means that (\ref{eq:eLP_objective}) is maximized as well. 

While there exist more filters than those that can be described by eLP, the combination of a weighted maximization problem subject to a non-redundancy constraint is a constructive way to describe meaningful filters. For example, it allows to generalize the notion of the outer-nodes filter to non-nested partitions. In addition, eLP can be more computationally efficient than Focused BH, which loops through every possible input e-value.

%% file: 08iii_appendix_global_partial.tex
For simplicity, consider first the case of a single level of resolution. We will therefore drop the subscript $m$ indicating the resolution $m \in \mathcal{M}$ in the following. 

Suppose that there are $L$ outcomes and consider $L \geq 1$ conditional independence hypotheses of the form 

\begin{equation}\label{eq:individual_resolution_outcome_cond_indep_null}
     H_{0}^{\ell} (A): Y_{\ell} \indep X_{A} \mid X_{- A}  \text{ versus } H_{1}^{\ell} (A): Y_{\ell} \not\!\perp\!\!\!\perp X_{A} \mid X_{- A}
\end{equation}

for each $\ell \in \{1, ..., L\}$ and each $A \in \mathcal{A}$. For simplicity of notation, we will refer to the hypothesis tested at location $A$ as $H_A^{\ell}$. Further, let $\mathcal{H}_0^{\ell} = \{H_A^{\ell}: H_{0}^{\ell} (A) \text{ is true} \}$ be the outcome-specific set of true null hypotheses. 

Let $v(A) \in \{1, ..., L\}$ denote the number of false conditional independence nulls (\ref{eq:individual_resolution_outcome_cond_indep_null}) at location $A \in \mathcal{A}$. Following \citet{benjamini2008screeningpartialconjunction}, we are interested in testing

\begin{equation}\label{eq:partial_conjunction_null}
  H_0^{u/L}(A): v(A) < u  \text{ versus } H_1^{u/L}(A): v(A) \geq u
 \end{equation}

for each $u \in \{1, ..., L\}$ and each $A \in \mathcal{A}$. We will refer to the partial conjunction hypothesis tested at location $A$ as $H_A^{u/L}$. Further, let $\mathcal{H}_0^{u/L} = \{H_A^{u/L}: H_0^{u/L}(A) \text{ is true} \} $ be the set of true partial conjunction nulls. For example, if $u = 1$, we are testing the global null hypothesis that a particular group $A \in \mathcal{A}$ is not conditionally associated to \textit{any} outcome.

We will first consider testing (\ref{eq:partial_conjunction_null}) with regular e-values in section \ref{subsec:partial_conjunction_reg_evalues}. In section \ref{subsec:partial_conjunction_knockoff_evalues} we consider e-values from the knockoff filter and discuss how to use them with KeLP.

\subsection{Partial conjunctions with regular e-values}\label{subsec:partial_conjunction_reg_evalues}

Consider testing (\ref{eq:partial_conjunction_null}) with regular e-values, as defined in section \ref{sec:elp}. The argument below closely follows \citet{benjamini2008screeningpartialconjunction}. We start by discussing how to combine $L$ regular e-values for outcome-specific hypotheses $H_A^\ell$, $\ell \in \{1, ..., L\}$, for group $A \in \mathcal{A}$ to obtain a valid partial conjunction e-value for group $A \in \mathcal{A}$. 

Following \citet{benjamini2008screeningpartialconjunction}, consider a vector of $L$ outcome-specific e-values at a specific location $A$ where the partial conjunction null holds, that is $A \in \mathcal{H}_0^{u/L}$. Further, for such a vector of e-values, let the first $1, \ldots, L-u+1$ e-values correspond to those where the outcome-specific null hypotheses are true. 
Then, let $U_1, ..., U_{L - u + 1}$ be e-values for $i \in \{1, ..., L - u + 1 \}$ and let $e_1, ... e_{u-1}$ be the other e-values. Let $e^{u/L} = f(U_1, ..., U_{L - u + 1}, e_1, ..., e_{u-1})$ be the e-value for the partial conjunction hypotheses using a combining method $f$. As long as $f$ is nondecreasing in all components, 

\begin{equation*}
    f(U_1, ..., U_{L - u + 1}, h_1(e_1), ..., h_{u-1}(e_{u-1})) \geq f(U_1, ..., U_{L - u + 1}, e_1, ..., e_{u-1})
\end{equation*}

will hold for functions $h_i(x) \geq x, i = 1, ..., u - 1$.

\begin{lemma}\label{lemma:flipped_partial_conjunction_lemma}
Under $H_0^{u/L}(A)$, let $h_i(e_i) \geq e_i$ for some function $h_i(\cdot)$, $i = 1, ..., u - 1$ and let 
\begin{eqnarray*} 
e_{*}^{u/L}(A) & = & f(U_1, ..., U_{L - u + 1}, h_1(e_1), ..., h_{u-1}(e_{u-1}))\\
e^{u/L}(A) & = & f(U_1, ..., U_{L - u + 1}, e_1, ..., e_{u-1})\end{eqnarray*} for a function $f$ nondecreasing in all components. Then $e_{*}^{u/L}(A) \succeq e^{u/L}(A)$.

\end{lemma}

Let $x \in [0, \infty]$ and define $\infty \geq \infty$.
If the event $\{e^{u/L}(A) \geq x \}$ occurs, then the event $\{f(U_1, ..., U_{L - u + 1}, h_1(e_1), ..., h_{u-1}(e_{u-1})) \geq x \}$ occurs, which proves the lemma as in \citet{benjamini2008screeningpartialconjunction}. 

Using the argument in \citet{benjamini2008screeningpartialconjunction}, under $H_0^{u/L}(A)$, the random variable $e^{u/L}(A) = f(U_1, ..., U_{L - u + 1}, e_1, ..., e_{u-1})$ with fixed $U_1, ..., U_{L - u + 1}$ satisfying $E(U_i) \leq 1$, is stochastically largest when each of the variables $e_1, ..., e_{u-1}$ is $\infty$ with probability one. Hence, we would like to make sure that in this case $\mathbb{E}[e_*^{u/L}(A)] \leq 1$, since then $1 \geq \mathbb{E}[e_*^{u/L}(A)] \geq \mathbb{E}[e^{u/L}(A)]$.

If $e_*^{u/L}(A)$ depends only on the $L - u + 1$ smallest e-values, for example by taking the average of the smallest $L - u + 1$ e-values, we would obtain $1 \geq \mathbb{E}[e_*^{u/L}(A)]$.

The pooled e-value would then be valid if it uses a combination function that satisfies: 

\begin{equation*}
    \mathbb{E} [f(U_1, ..., U_{L - u + 1}, \infty, ..., \infty)] \leq 1 
\end{equation*}

In summary, for regular e-values, we could just take the average of the smallest $L - u + 1$ e-values to test the partial conjunction hypothesis at a group $A$. 

Intuitively, this matches the result of \citet{benjamini2008screeningpartialconjunction} for p-values, who focus on the largest $L-u+1$ p-values to find partial conjunction p-values. While other e-value combination methods could also be applied, such as a Simes e-value, we choose the arithmetic mean since it essentially dominates any symmetric e-merging function for the global null as shown by \citet{vovk2021valuescalibrationcombination}. Note also that the average e-value is an e-value regardless of the dependency structure between the e-values by linearity of expectation; see also \citet{vovk2021valuescalibrationcombination}.

The partial conjunction e-values stemming from regular e-values discussed in this section can be directly used in eLP; see section \ref{sec:elp}. We next discuss how to construct partial conjunction e-values based on the knockoff filter.

\subsection{Partial conjunction hypotheses with knockoff e-values}\label{subsec:partial_conjunction_knockoff_evalues}

We first show that partial conjunction e-values derived from the knockoff filter are not strictly e-values in the sense that the expected value under the null is bounded by 1 for every individual partial conjunction e-value, but instead fulfill a more relaxed condition. Then we consider how to use these knockoff partial conjunction e-values in KeLP.

To derive conditions for partial conjunction e-values from the knockoff filter, for simplicity of notation, we first focus on a single level of resolution in theorem \ref{thm:partial_conjunction_knockoff_evalue_condition}. We therefore first drop the subscript indicating the level of resolution - as such, in theorem \ref{thm:partial_conjunction_knockoff_evalue_condition}, $c$ replaces $c_m$ in equation (\ref{eq:main_evalue_definition}).

\begin{theorem}\label{thm:partial_conjunction_knockoff_evalue_condition}

  Given a single-resolution partition $\mathcal{A}$, we consider the set of partial conjunction hypotheses $\mathcal{H}^{u/L}$ as a family of hypotheses addressed by KeLP. We assume that valid $e_A^{\ell}$ fulfilling the conditions described in Procedure \ref{def:lp_KeLP} are available for every $A \in \mathcal{A}$ and any $\ell \in \{1, ..., L\}$. For any $A \in \mathcal{A}$, let $\bar{e}^{u/L}_{A} := \frac{1}{L - u + 1} \sum_{i = 1}^{L - u + 1} e_{A}^{[u - 1 + i]}$ , where $e_{A}^{[1]}, ..., e_{A}^{[L]}$ denote the sorted knockoff e-values from largest to smallest. Let $c$ denote the fixed multiplier as in the definition of knockoff e-values in equation (\ref{eq:main_evalue_definition}). Then these partial conjunction e-values fulfill
    
    \begin{equation*}
        \sum_{A: H_A^{u/L} \in \mathcal{H}_0^{u/L}} \mathbb{E} \Big[\bar{e}^{u/L}_{A} \Big] \leq  c \times \frac{L}{L - u + 1}
    \end{equation*}

\end{theorem}

 \begin{proof}

First, consider the case where $u = 1$. Recall that by section \ref{sec:kelp}, for each outcome $\ell \in \{1, ..., L\}$, the knockoff e-values fulfill $\sum_{A: H_A^{\ell} \in \mathcal{H}_0^{\ell}} E[e^{\ell}_A] \leq c$. Then

    \begin{align*}
        \sum_{A: H_A^{u/L} \in \mathcal{H}_0^{u/L}} \mathbb{E}[\bar{e}_{A}^{u/L} ] &= 
        \sum_{A: H_A^{u/L} \in \mathcal{H}_0^{u/L}} \mathbb{E}[\frac{1}{L} \sum_{\ell = 1}^L e^{\ell}_{A}] \\
        &= \frac{1}{L} \sum_{\ell = 1}^L \sum_{A: H_A^{u/L} \in \mathcal{H}_0^{u/L}} \mathbb{E}[e^{\ell}_{A}] \\
       &\leq \frac{1}{L} \sum_{\ell = 1}^L \sum_{A: H_A^{\ell} \in \mathcal{H}_0^{l}} \mathbb{E}[e^{\ell}_{A}]  \\
       &\leq \frac{1}{L} \sum_{\ell = 1}^L c \\
       &= c
    \end{align*}

The first inequality follows as $\mathcal{H}_0^{u/L} \subseteq \mathcal{H}_0^{l}$ for each $\ell \in \{1, ..., L\}$. Note that if $H_A^{u/L} \in \mathcal{H}_0^{u/L}$, then $v(A) = 0$ and $H_A^{\ell} \in \mathcal{H}_0^{l}$ for every $\ell \in \{1, ..., L\}$. However, $H_A^{\ell} \in  \mathcal{H}_0^{l}$ does not imply that $H_A^{u/L} \in \mathcal{H}_0^{u/L}$ if $v(A) \geq 1$ for at least one $\ell \in \{1, ..., L \}$. The last inequality follows since $\sum_{A: H_A^{\ell} \in \mathcal{H}_0^{l}} \mathbb{E} \Big[  e^{\ell}_{A} \Big] \leq c$ by construction of the knockoff e-values.

Next, consider the case that $u > 1$. The following two observations are crucial here: (1) For every $A: H_{A}^{u/L} \in \mathcal{H}_0^{u/L}$, the outcomes $\ell \in \{1, ..., L\}$ corresponding to the $L - u + 1$ considered e-values might be different. (2) The condition $\sum_{A: H_A^\ell \in \mathcal{H}_0^{\ell}} E[e^{\ell}_A] \leq c$ holds for every outcome separately. 

As observation (2) concerns the sum over the expected null e-values for any outcome $\ell \in \{1, ..., L\}$, it could happen, separately for each outcome $\ell \in \{1, ..., L\}$, that $\mathbb{E}[e_{A'}^{\ell}] = c$ for some $A': H_{A'}^{\ell} \in \mathcal{H}_0^{\ell}$, which implies that $\sum_{A: H_A^{\ell} \in \mathcal{H}_0^{\ell} \setminus H_{A'}^{\ell}} \mathbb{E}[e_{A}^{\ell}] = 0$. 

Recall that as we only consider a single level of resolution for simplicity here, we let $\mathcal{A}$ denote the partition of $\{1, \ldots,p\}$ into disjoint groups. With $A_{g} \in \mathcal{A}$ for $g \in \{1, ..., |\mathcal{A}|\}$ we indicate the elements of the partition $\mathcal{A}$, each of which is a disjoint set of indices $j \in \{1, ..., p\}$. This partition $\mathcal{A}$ is the same for every outcome $\ell \in \{1, ..., L\}$. Below we find it helpful to consider different groups for different outcomes $\ell \in \{1, ..., L\}$. These groups will be denoted by $A_{g^\ell} \in \mathcal{A}$. 

Based on the discussion above, consider the case that for every $\ell \in \{1, ..., L\}$ there is such a group $A'_{g^\ell} \in \mathcal{A}$ with $H_{A'_{g^\ell}}^{\ell} \in \mathcal{H}_0^{\ell}$ where $\mathbb{E}[e_{A'_{g^\ell}}^{\ell}] = c$.

First suppose that $A'_{g} = A'_{g^1} = ... = A'_{g^L}$ for $A'_{g} \in \mathcal{A}$. As such, there is a single location $A'_{g} \in \mathcal{A}$ where for every $\ell \in \{1, ..., L\}$, $H_{A'_{g}}^{\ell} \in \mathcal{H}_0^{\ell}$ and $\mathbb{E}[e_{A'_g}^{\ell}] = c$. Suppose that $A'_{g}$ is also part of a partial conjunction null: $H_{A'_{g}}^{u/L} \in \mathcal{H}_0^{u/L}$. By construction of the partial conjunction e-values, this implies that $\bar{e}^{u/L}_{A'_{g}} = c$. Recall that for each outcome $\ell \in \{1, ..., L\}$ there can be only a single location $A'_{g^{\ell}}$ with $\mathbb{E}[e_{A'_{g^{\ell}}}^{\ell}] = c$ for $H_{A'_{g^\ell}}^{\ell} \in \mathcal{H}_0^{\ell}$, which implies that $\sum_{A_{g^{\ell}}: H_{A_{g^\ell}}^{\ell} \in \mathcal{H}_0^{\ell} \setminus H_{A'_{g^\ell}}^{\ell}} \mathbb{E}[e_{A_{g^{\ell}}}^\ell] = 0$. Therefore, $\sum_{A: H_A^{u/L} \in \mathcal{H}_0^{u/L}} \mathbb{E} \Big[\bar{e}^{u/L}_{A} \Big] \leq c$ is still fulfilled in this case.

Next, consider the case that $A'_{g^1} \neq ... \neq A'_{g^L}$. This means that the locations $A'_{g^{\ell}} \in \mathcal{A}$ with $H_{A'_{g^\ell}}^{\ell} \in \mathcal{H}_0^{\ell}$ and $\mathbb{E}[e_{A'_{g^\ell}}^{\ell}] = c$ are different for every $\ell \in \{1, ..., L\}$.

Recall observation (1) and suppose that the partial conjunction null is true at every location $A'_{g^1}, ..., A'_{g^L}$, i.e. $\{H_{A'_{g^1}}^{u/L}, ..., H_{A'_{g^L}}^{u/L}\} \subseteq \mathcal{H}_0^{u/L}$. By the same logic as above, we obtain that $\bar{e}^{u/L}_{A'_{g^1}} = ... = \bar{e}^{u/L}_{A'_{g^L}} = c \times \frac{1}{L - u + 1}$. 

As there can be at most one group $A'_{g^\ell}$ with $\mathbb{E}[e_{A'_{g^\ell}}^{\ell}] = c$ for every $\ell \in \{1, ..., L\}$, we can bound $\sum_{A: H_A^{u/L} \in \mathcal{H}_0^{u/L}} \mathbb{E} \Big[\bar{e}^{u/L}_{A} \Big] \leq \min \Big\{c \times (\frac{L}{L - u + 1}), c \times (\frac{c}{L - u + 1}) \Big\}$. In most applications $c \geq L$. Either way, we can always apply the following bound: 

\begin{equation*}
     \sum_{A: H_A^{u/L} \in \mathcal{H}_0^{u/L}} \mathbb{E} \Big[\bar{e}^{u/L}_{A} \Big] \leq  c \times \frac{L}{L - u + 1}
\end{equation*}

\end{proof}

In order to obtain FDR control with partial conjunction e-values with KeLP, it would be required that $\sum_{A: H_A^{u/L} \in \mathcal{H}_0^{u/L}} \mathbb{E} \Big[\bar{e}^{u/L}_{A} \Big] \leq  |\mathcal{A}|$; see section \ref{sec:kelp}. By theorem \ref{thm:partial_conjunction_knockoff_evalue_condition}, for FDR control it is therefore required that $c \times [L/(L - u + 1)] \leq  |\mathcal{A}|$. Expressed differently, if only $c \leq |\mathcal{A}|$ is ensured, then FDR would not be controlled at level $\alpha$, but instead at level $\alpha \frac{L}{L - u + 1}$. 

We next extend the previous result for a single level of resolution to multiple levels of resolution. 

\begin{ProcedureDef}{\em \textbf{(KeLP with partial conjunction e-values)}}\label{def:lp_KeLP_partial_conjunction}
Let ${\cal H}^{u/L}$ be a multi-resolution family of partial conjunction hypotheses addressed with KeLP. For each resolution $m\in {\cal M}$ and for each outcome $\ell \in \{1, ..., L\}$, construct (group) knockoffs as defined in section \ref{sec:kelp} and obtain the corresponding e-value $e_{A_g^m}^{\ell}$ for every $A_g^m \in \mathcal{A}^m$ and every $\ell \in \{1, ..., L\}$ as in Procedure \ref{def:lp_KeLP}, letting $\{c_m\}_{m \in \mathcal{M}}$ be non negative numbers such that $\frac{L}{L - u + 1} \sum_{m \in \mathcal{M}} c_m \leq |\mathcal{A}|$. Then, within each level of resolution $m \in \mathcal{M}$, construct a partial conjunction e-value $\bar{e}^{u/L}_{A_g^m}$ for every $A_g^m \in \mathcal{A}^m$ as described in theorem \ref{thm:partial_conjunction_knockoff_evalue_condition}. 
For a fixed weighting function $w(\cdot)$, and a target FDR level $\alpha$ the rejection set of KeLP with partial conjunction e-values is given by the output of eLP (Procedure \ref{def:lp_eLP}) taking the partial conjunction e-values as input.
\end{ProcedureDef}

At each level of resolution separately, the partial conjunction e-values fulfill the property that 

 \begin{equation*}
        \sum_{A_g^m: H_{g}^{m, u/L} \in \mathcal{H}_0^{u/L, m}} \mathbb{E} \Big[\bar{e}^{u/L}_{A_g^m} \Big] \leq  c_m \frac{L}{L - u + 1}.
    \end{equation*}

Across all resolutions, this implies that 

\begin{equation*}
    \sum_{m \in \mathcal{M}}  \sum_{A_g^m : H_{g}^{u/L, m} \in \mathcal{H}_0^{u/L, m}} \mathbb{E} \Big[\bar{e}^{u/L}_{A_g^m} \Big] \leq \ \frac{L}{L - u + 1} \sum_{m \in \mathcal{M}} c_m.
\end{equation*}

By theorem \ref{thm:main_thm_elp} and the same argument as in section \ref{sec:kelp}, for FDR control at level $\alpha$ it is therefore required to set $c_m$ such that 

\begin{equation*}\label{eq:knockoff_partial_conjunction_evalue_fulfill}
 \frac{L}{L - u + 1} \sum_{m \in \mathcal{M}} c_m \leq |\mathcal{A}|.
\end{equation*} 

Alternatively, if it is only ensured that $\sum_{m \in \mathcal{M}} c_m \leq |\mathcal{A}|$, the FDR would only be controlled at level $\alpha \frac{L}{L - u + 1}$.

However, note that in practice this adjustment factor ($L / (L - u + 1)$) for knockoff e-values might not be needed, as the existence of groups $A$ where the partial conjunction null holds with $\mathbb{E}[e_{A}^{\ell}] = c$ and $A \in \mathcal{H}_0 ^{\ell}$ (which implies that $\mathbb{E}[e_{A'}^{\ell}] = 0$ for all other $A' \in \mathcal{H}_0 ^{\ell} \setminus A$) that are distinct across $\ell \in \{1, ..., L\}$ seems rather unlikely.

%% file: 08iv_appendix_emlkf.tex
\subsection{Description}

The p-filter by \citet{barber2017pfilter} coordinates rejections across partitions and controls the FDR simultaneously for all partitions. The original p-filter by \citet{barber2017pfilter} has been extended to arbitrary dependencies between the p-values in \citet{pfilter2019} using reshaping, which makes the procedure conservative. As mentioned in \citet{wang2022eBH}, it is possible to develop an analog of the p-filter for e-values.

The notation follows \citet{barber2017pfilter} and \citet{wang2022eBH}. Suppose that we have $n$ regular, individual-level e-values $\mathbf{e} = (e_1, ..., e_n)$ corresponding to hypotheses $H_1, ..., H_n$ and an unknown set of nulls $\mathcal{H}_0$. Suppose that the hypotheses are organized into $|\mathcal{M}|$ different layers, where each layer consists of an arbitrary partition of the hypotheses into groups. Note that we are leaving the setting described in section \ref{sec:introduction}. The hypotheses can be arbitrary and the e-values $\mathbf{e} = (e_1, ..., e_n)$ do not have to correspond to features, which is why we will use the index notation $i \in \{1, ..., n\}$ as opposed to $j \in \{1, ..., p\}$. Moreover, the partitions do not have to be resolutions. We will return to the previous setting based on knockoff e-values and contiguous groups of different sizes with the same feature basis at the end of this section. 

We use the arithmetic mean to combine e-values within each group as it dominates the analogue of Simes \citep{simes1986improved} for e-values \citep{vovk2021valuescalibrationcombination}. However, if other group-level e-values are already available, these can be used instead as well. At the end of this section, we describe how to coordinate discoveries by the knockoff filter based on group-level knockoff e-values. As before, let $A_{g}^m \in \mathcal{A}^m$ for $g \in \{1, ..., |\mathcal{A}^m|\}$ be the set of indices $i \in \{1, ..., n\}$ belonging to group $A_{g}^m$.

We define $\infty \times \text{constant} = \infty$ and $\infty \geq \infty$. In this section, we consider regular e-values as defined in section \ref{sec:elp}. Note that by definition, e-values can take on the value $\infty$; see \citet{vovk2021valuescalibrationcombination}. However, in the context of knockoff e-values, as described in section \ref{sec:kelp}, $\infty$ is not a possible value.

We follow the notation and problem setup in \citet{barber2017pfilter}. Consider thresholds $(t_1, ..., t_{|\mathcal{M}|}) \in [0, \infty]^{|\mathcal{M}|}$. Assume for now that the thresholds are given, following \citet{barber2017pfilter}, we will show below how to define these thresholds. Define the set of all discoveries that are made by the algorithm as:

\begin{equation*}
   \mathcal{R} (t_1, ..., t_{|\mathcal{M}|}) :=  \{i: \text{for all $m$, Mean}(e_{A^m_{g(i)}}) \geq t_m  \}, 
\end{equation*}

where $g(i)$ is the index of the group that $e_i$ belongs to in partition $m$. Thus, a hypothesis $i$ is selected if at every layer $m$, its group $A^m_{g(i)}$ has $\text{Mean}(e_{A^m_{g(i)}}) \geq t_m$.

Correspondingly, the rejection set at layer $m$ is defined as 

\begin{equation*}
   \mathcal{R}_m (t_1, ..., t_{|\mathcal{M}|}) :=   \{A_g^m \in \mathcal{A}^m: \mathcal{R}(t_1, ..., t_{|\mathcal{M}|}) \cap A_g^m\neq \emptyset \}.
\end{equation*}

Further define:

\begin{equation*}
    \widetilde{\mathcal{R}}_m(t_1, ..., t_{|\mathcal{M}|}) := | \mathcal{R}_m (t_1, ..., t_{|\mathcal{M}|}) | \vee 1
\end{equation*}

Following \citet{barber2017pfilter}, let

\begin{equation}\label{eq:efilter_threshold_set}
   \hat{\mathcal{T}}(\alpha_1, ..., \alpha_{|\mathcal{M}|}) := \{(t_1, ..., t_{|\mathcal{M}|}) \in [0, \infty]^{|\mathcal{M}|} : t_m  \widetilde{\mathcal{R}}_m(t_1, ..., t_{|\mathcal{M}|}) \geq |\mathcal{A}^m| / \alpha_m  \text{ for all $m$ } \} 
\end{equation}

where $\alpha_1, ..., \alpha_{|\mathcal{M}|} \in [0, 1]$ is the desired level of FDR control for each layer $m$. Note that this ensures that the FDR is controlled at each level $m$ by self-consistency; see theorem \ref{thm:main_thm_elp}. An alternative proof is included below theorem \ref{thm:efilter_fdr_control}.

Continuing with the argument in \citet{barber2017pfilter}, we let 

\begin{equation}\label{eq:efilter_individual_threshold}
    \hat{t}_m := \min \{t_m \in [0, \infty]: \exists t_1, ..., t_{m-1}, t_{m+1}, ..., t_{|\mathcal{M}|} \text{ such that } (t_1, ..., t_{|\mathcal{M}|}) \in \hat{\mathcal{T}}(\alpha_1, ..., \alpha_{|\mathcal{M}|}) \}
\end{equation}

for each $m = 1, ..., |\mathcal{M}|$.

\begin{theorem}\label{thm:efilter_threshold_has_max}

Fix $\alpha_1, ..., \alpha_m \in [0, 1]$ and a vector of e-values $\mathbf{e} \in [0, \infty]^n$. Then $(\hat{t}_1, ..., \hat{t}_m) \in \hat{\mathcal{T}}(\alpha_1, ..., \alpha_{|\mathcal{M}|})$.

\end{theorem}

The proof follows \citet{barber2017pfilter}:

\begin{proof}

By definition of $\hat{t}_m$ there is some $t_1^m, ..., t_{m-1}^m, t_{m+1}^m, ..., t_{|\mathcal{M}|}^m$ for each $m$ such that 

\begin{equation*}
    (t_1^m, ..., t_{m-1}^m, \hat{t}_m ,t_{m+1}^m, ..., t_{|\mathcal{M}|}^m) \in \hat{\mathcal{T}}(\alpha_1, ..., \alpha_{|\mathcal{M}|}).
\end{equation*}

Therefore

\begin{equation*}
    \mathcal{R}(t_1^m, ..., t_{m-1}^m, \hat{t}_m, t_{m+1}^m, ..., t_{|\mathcal{M}|}^m) \subseteq \mathcal{R}(\hat{t}_1, ..., \hat{t}_{m-1}, \hat{t}_m, \hat{t}_{m+1}, ..., \hat{t}_{|\mathcal{M}|})
\end{equation*}

because $\mathcal{R}(t_1, ..., t_{|\mathcal{M}|})$ is a non-increasing function of $(t_1, ..., t_{|\mathcal{M}|})$ and for each $m' \neq m$, $\hat{t}_{m'} \leq t_{m'}^m$.

By definition, 
$\hat{t}_m \widetilde{\mathcal{R}}_m(\hat{t}_1^m, ..., \hat{t}_{m-1}^m, \hat{t}_m, \hat{t}_{m+1}^m, ..., \hat{t}_{|\mathcal{M}|}^m) \geq |\mathcal{A}^m| / \alpha_m$. 

But by the above we have that

\begin{equation*}
    \widetilde{\mathcal{R}}(t_1^m, ..., t_{m-1}^m, \hat{t}_m, t_{m+1}^m, ..., t_{|\mathcal{M}|}^m) \leq \widetilde{\mathcal{R}}(\hat{t}_1, ..., \hat{t}_{m-1}, \hat{t}_m, \hat{t}_{m+1}, ..., \hat{t}_{|\mathcal{M}|})
\end{equation*}

So therefore, we have that 

\begin{align*}
    \hat{t}_m \widetilde{\mathcal{R}}(\hat{t}_1, ..., \hat{t}_{m-1}, \hat{t}_m, \hat{t}_{m+1}, ..., \hat{t}_{|\mathcal{M}|}) &\geq \hat{t}_m \widetilde{\mathcal{R}}(t_1^m, ..., t_{m-1}^m, \hat{t}_m, t_{m+1}^m, ..., t_{|\mathcal{M}|}^m) \\
    &\geq |\mathcal{A}^m| / \alpha_m
\end{align*}

Since this holds for all $m$, this shows that $(\hat{t}_1, ..., \hat{t}_m) \in \hat{\mathcal{T}}(\alpha_1, ..., \alpha_{|\mathcal{M}|})$ by definition of $ \hat{\mathcal{T}}(\alpha_1, ..., \alpha_{|\mathcal{M}|})$.

\end{proof}

\begin{theorem}\label{thm:efilter_fdr_control}
Let $(\hat{t}_1, ..., \hat{t}_m)$ as above. Then, for each $m = 1, ..., |\mathcal{M}|$ the method controls the FDR for the $m$-th partition, 

\begin{equation*}
    \mathbb{E}[FDP_m(\hat{t}_1, ..., \hat{t}_{|\mathcal{M}|})] \leq \alpha_m \frac{|\mathcal{H}_0^m|}{|\mathcal{A}^m|}.
\end{equation*}

\end{theorem}

\begin{proof}

As with the corresponding p-value thresholds, $t_1, ..., t_{|\mathcal{M}|}$ take values in

\begin{equation*}
    \hat{t}_m \in \Big\{ \big\{ \frac{1}{\alpha_m} \frac{|\mathcal{A}^m|}{k}  : k = 1, ..., |\mathcal{A}^m| \big\} \cup \{ \infty \} \Big\} . 
\end{equation*}

Fix any partition $m$. We will show FDR control separately for the case $\hat{t}_m = \infty$ and $\hat{t}_m < \infty$.

First suppose that $\hat{t}_m = \infty$.

Let $e_{A, avg}^m$ denote the average of the individual-level e-values for a particular group $A_g^m$. If $A_g^m \in  \mathcal{R}_m (\hat{t}_1, ..., \hat{t}_{|\mathcal{M}|})$, then we must have $e_{A, avg}^m \geq \hat{t}_m$. 

Fix any $A_g^m: H_g^m \in \mathcal{H}_0^m$. If $\hat{t}_m = \infty$, then $\mathbf{1} \{A_g^m \in  \mathcal{R}_m (\hat{t}_1, ..., \hat{t}_{|\mathcal{M}|}) \} = 0$ for any $A_g^m: H_g^m \in \mathcal{H}_0^m$ almost surely.

This is because we know that $\mathbb{E}[e_{A_g^m}] \leq 1$ for all $A_g^m: H_g^m \in \mathcal{H}_0^m$. However, this implies also that $e_{A_g^m} < \infty$ almost surely, because if $\infty$ was a possible value, then $\mathbb{E}[e_{A_g^m}] = \infty$, which cannot be true since $\mathbb{E}[e_{A_g^m}] \leq 1$ .

Then

\begin{align*} 
 FDP_m(\hat{t}_1, ..., \hat{t}_{m-1}, \infty, \hat{t}_{m+1}, ..., \hat{t}_{|\mathcal{M}|}) &= \frac{|\mathcal{R}_m(\hat{t}_1, ..., \hat{t}_{|\mathcal{M}|}) \cap \mathcal{H}_0^m|}{1 \vee |\mathcal{R}_m(\hat{t}_1, ..., \hat{t}_{|\mathcal{M}|})|}  \\ 
    &=  \sum_{A_g^m: H_g^m \in \mathcal{H}_0^m} \frac{\mathbf{1} \{A_g^m \in  \mathcal{R}_m (\hat{t}_1, ..., \hat{t}_{|\mathcal{M}|} )\}}{1 \vee |\mathcal{R}_m(\hat{t}_1, ..., \hat{t}_{|\mathcal{M}|})|} \\
     & = 0
\end{align*}

And therefore we also have that

\begin{align*} 
 \mathbb{E}[FDP_m(\hat{t}_1, ..., \hat{t}_{m-1}, \infty, \hat{t}_{m+1}, ..., \hat{t}_{|\mathcal{M}|})]
  &=  \mathbb{E}[0]  \\ 
  &= 0 \\ 
  &\leq \alpha_m \frac{|\mathcal{H}_0^m|}{|\mathcal{A}^m|} 
\end{align*}

Now suppose that $\hat{t}_m < \infty$.

The above argument implies that if $A_g^m \in  \mathcal{R}_m (\hat{t}_1, ..., \hat{t}_{|\mathcal{M}|})$ for some null group $A_g^m: H_g^m \in \mathcal{H}_0^m$, then we must have $\hat{t}_m < \infty$ almost surely. Moreover, since $\hat{t}_m < \infty$ by the construction above we have that $\hat{t}_m \in [\frac{1}{\alpha}_m, \frac{|\mathcal{A}^m|}{\alpha_m}]$, so $\hat{t}_m > 0$.

Moreover, note that $t_m R(t) / |\mathcal{A}^m|$ only has downside jumps and if $t_m = 0$ this becomes zero; see also \citet{wang2022eBH}. Therefore 

\begin{equation*}
    \hat{t}_m \widetilde{\mathcal{R}}_m(\hat{t}_1, ..., \hat{t}_{|\mathcal{M}|}) = \frac{|\mathcal{A}^m|}{\alpha_m}
\end{equation*}

Thus,

\begin{align*} 
 FDP_m(\hat{t}_1, ..., \hat{t}_{|\mathcal{M}|}) &= \frac{|\mathcal{R}_m(\hat{t}_1, ..., \hat{t}_{|\mathcal{M}|}) \cap \mathcal{H}_0^m|}{1 \vee |\mathcal{R}_m(\hat{t}_1, ..., \hat{t}_{|\mathcal{M}|})|}  \\ 
 &= \frac{|\mathcal{R}_m(\hat{t}_1, ..., \hat{t}_{|\mathcal{M}|}) \cap \mathcal{H}_0^m|}{\widetilde{\mathcal{R}}_m(\hat{t}_1, ..., \hat{t}_{|\mathcal{M}|})}\\
    &=  \sum_{A_g^m: H_g^m \in \mathcal{H}_0^m} \frac{\mathbf{1} \{A_g^m \in  \mathcal{R}_m (\hat{t}_1, ..., \hat{t}_{|\mathcal{M}|}) \}}{\widetilde{\mathcal{R}}_m(\hat{t}_1, ..., \hat{t}_{|\mathcal{M}|})} \\
     &\leq \alpha_m \frac{1}{|\mathcal{A}^m|}  \hat{t}_m \sum_{A_g^m: H_g^m \in \mathcal{H}_0^m} \mathbf{1} \{A_g^m \in  \mathcal{R}_m (\hat{t}_1, ..., \hat{t}_{|\mathcal{M}|}) \}
\end{align*}

Fix any null group $A_g^m: H_g^m \in \mathcal{H}_0^m$. If $A_g^m$ is selected, then we must have $e_{A, avg}^m \geq \hat{t}_m$.

Recall that $A_g^m$ is a null group. Since $\mathbb{E}[e_i] \leq 1$ for each $i \in A_g^m$, we also have that $\mathbb{E}[e_{A, avg}^m] \leq 1$. Note also both $e_{A, avg}^m$ and $\hat{t}_m$ are non-negative.

\begin{align*} 
 FDP_m(\hat{t}_1, ..., \hat{t}_{|\mathcal{M}|}) &\leq \alpha_m \frac{1}{|\mathcal{A}^m|}  \hat{t}_m \sum_{A_g^m: H_g^m \in \mathcal{H}_0^m} \mathbf{1} \{A_g^m \in  \mathcal{R}_m (\hat{t}_1, ..., \hat{t}_{|\mathcal{M}|}) \} \\
 & = \alpha_m \frac{1}{|\mathcal{A}^m|}  \hat{t}_m \sum_{A_g^m: H_g^m \in \mathcal{H}_0^m} \mathbf{1} \{e_{A, avg}^m \geq \hat{t}_m \} 
\end{align*}

Then, taking expectations we get that

\begin{align*} 
 \mathbb{E}[FDP_m(\hat{t}_1, ..., \hat{t}_{|\mathcal{M}|})]
  &\leq  \alpha_m \frac{1}{|\mathcal{A}^m|} \sum_{A_g^m: H_g^m \in \mathcal{H}_0^m}  \mathbb{E} \big[ \hat{t}_m \mathbf{1} \{e_{A, avg}^m \geq \hat{t}_m  \}  \big] \\
  &\leq \alpha_m \frac{1}{|\mathcal{A}^m|} \sum_{A_g^m: H_g^m \in \mathcal{H}_0^m}  \mathbb{E} \big[ e_{A, avg}^m \big] \\
  &\leq \alpha_m \frac{1}{|\mathcal{A}^m|}  \sum_{A_g^m: H_g^m \in \mathcal{H}_0^m}  (1)  \\
   &\leq \alpha_m  \frac{|\mathcal{H}_0^m|}{|\mathcal{A}^m|} 
\end{align*}

Hence, we get FDR control in either case: if $\hat{t}_m = \infty$ and if $\hat{t}_m < \infty$.

\end{proof}

To find the thresholds $\hat{t}_m$, the algorithm in \citet{barber2017pfilter} can be adapted:

\begin{algorithm}[H]\label{alg:filter_find_thresholds}
\caption{Find thresholds $\hat{t}_m$, following \citet{barber2017pfilter} }\label{alg:efilter_thresholds}
\KwIn{A vector of e-values $\mathbf{e} \in [0, \infty]^n$, target FDR levels $\alpha_1, .., \alpha_{|\mathcal{M}|}$, arbitrary partitions $\mathcal{A}^m$ of the hypotheses into groups for each $m \in \mathcal{M}$}
\KwOut{Thresholds $\hat{t}_m$ for $m \in \mathcal{M}$}
Initialize: $t_1 \gets 1/\alpha_1, ..., t_{|\mathcal{M}|} = 1/\alpha_{|M|}$\;
\While{$t_1, ..., t_{|\mathcal{M}|}$ are all unchanged in the last round}{
\For{$m = 1, ..., |\mathcal{M}|$}{
    Define $ \mathcal{R}_m(t_1, ..., t_{|\mathcal{M}|}) = \{A_g^m \in  \mathcal{A}^m\}: \mathcal{R}(t_1, ..., t_{|\mathcal{M}|}) \cap A_g^m \neq \emptyset \}$ (as above)\;
     $t_m \gets \min \Bigl\{ \widetilde{t} \in [t_m, \infty] :   \widetilde{t} \widetilde{\mathcal{R}}_m(t_1, ..., t_{m-1}, \widetilde{t}, t_{m+1}, ..., t_{|\mathcal{M}|}) \geq |\mathcal{A}^m| / \alpha_m   \Bigl\}$
    
}
}
\end{algorithm}

This construction can also be used to obtain a Multilayer Knockoff Filter inspired by \citet{katsevich2019multilayer} with valid FDR control at each level $m \in \mathcal{M}$. To do so, we define e-values as in section \ref{sec:kelp} to test the conditional independence hypothesis at each level of resolution for a single outcome. However, it is important to set $c_m = |\mathcal{A}^m|$ in order to obtain FDR control at each level $m \in \mathcal{M}$. This ensures that 

\begin{equation*}
     \sum_{A_{g}^m : H_g^m \in \mathcal{H}_0^{m}} \mathbb{E} \Big[  e_{A_{g}^m} \Big] \leq |\mathcal{A}^m|
\end{equation*}

holds at each level of resolution separately. As shown in theorem \ref{thm:main_thm_elp}, this condition is sufficient to obtain FDR control in combination with self-consistency. As the thresholds $t_1, ..., t_{|\mathcal{M}|}$ are obtained as described above; see equation (\ref{eq:efilter_threshold_set}), the knockoff e-value based multilayer knockoff filter also controls the FDR at each level $m \in \mathcal{M}$. 

Note that as in section \ref{sec:kelp}, $W_m$ and $T_m$ are defined separately for each layer. Moreover, the coordination of discoveries is conducted via the thresholds $t_1, ..., t_{|\mathcal{M}|}$ (\ref{eq:efilter_threshold_set}). This is in contrast to \citet{katsevich2019multilayer}, who coordinate the rejections in different layers directly via the knockoff stopping time (\ref{eq:stopping_Mtime_gamma}) and control the FDR defined by the knockoff \citep{candes2018panning, barber2015controlling}. If we were to translate this procedure into e-values, the procedure by \citet{katsevich2019multilayer} would be based on
        
\begin{equation}\label{eq:emlkf_evalue}
    \widetilde{e}_{A_{g}^m} = |\mathcal{A}^m| \cdot \frac{\mathbf{1} \{W_{A_{g}^m} \geq T^*_m \}}{1 + \sum_{A \in \mathcal{A}^m} \mathbf{1} \{W_A \leq -T^*_m\}}
\end{equation}

with $T^*_m$ defined as in \citet{katsevich2019multilayer} in combination with the e-BH procedure at level $\alpha_m$ at each layer $m \in \mathcal{M}$. In equation (\ref{eq:emlkf_evalue}), $T_m^*$ not only depends on $m \in \mathcal{M}$, but also on the other layers $m' \neq m \in \mathcal{M}$; see \citet{katsevich2019multilayer} for further details. The difficulty with using $T^*_m$ dependent on other layers is that $T^*_m$ cannot be represented as a stopping time that only depends on $W_m$ anymore. This results in FDR control at level $\alpha_m \times \kappa$ instead of level $\alpha_m$, where $\kappa \approx 1.93$.

\subsection{Simulation}

For comparability, we follow the simulation setting in \citet{katsevich2019multilayer} based on $M = 2$ layers. Layer 1 is comprised of individual variables and layer 2 of groups of size 10 (i.e. each group contains 10 individual variables). 

We simulate data from a linear model $ Y = X \mathbf{\beta} + \mathbf{\epsilon}$, where $\mathbf{\epsilon} \sim \mathcal{N}(0, I)$ with $n = 4,500$ and $p = 2,000$. $X \in \mathbb{R}^{p \times n}$ is sampled independently from $\mathcal{N}(0, \Sigma_{\rho})$ for $\rho = 0.3$ where $(\Sigma_{\rho})_{ij} = \rho^{|i-j|}$ is the covariance of an AR(1) process with correlation $\rho$.

Following \citet{katsevich2019multilayer}, we first select $r \in \{10, 20\}$ groups uniformly at random and then choose uniformly at random 75 elements of these $r$ groups to determine the nonzero elements of $\beta$. The absolute value of the 75 non-zero elements $\beta$ are determined by $a / \sqrt{n}$, where $a$ denotes the signal amplitude, which is the parameter that we will vary. The sign of the non-zero coefficient is determined by independent coin flips. Following \citet{katsevich2019multilayer}, we construct fixed-equi knockoff variables and choose lasso-based variable importance statistics combined using the signed-max function. We tune $\gamma$ on an independently generated tuning data set with the same parameters and dimensions. Each point in the simulation results below is based on 100 iterations.

We compare the multilayer knockoff filter (MKF) with the knockoff e-value version (e-MKF) at the same theoretical guarantee of $\alpha_m = 0.2$ for both $m \in \{1, 2\}$. This shows that the e-MKF has higher power compared to the MKF for the same theoretical level of FDR control. 

\foreach \k in {10, 20}
{

    \foreach \m in {mult}
    { 

    \foreach \g in {tuned} 
    {

\begin{figure}[H]
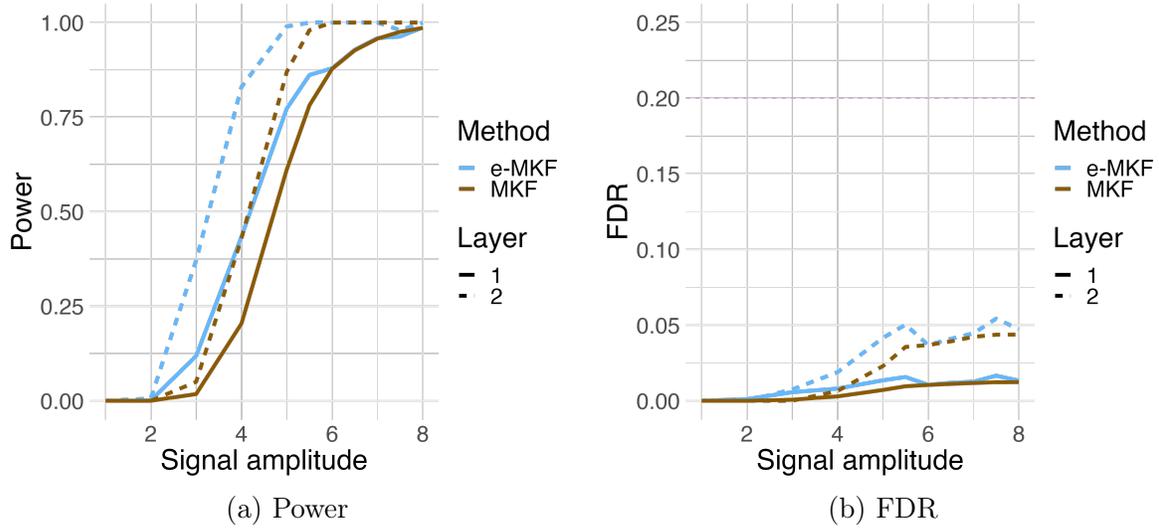

\begin{tabular}{cc}
  \includegraphics[width=75mm]{figures/emlkf/emlkf_power_nrep100_a_0.2ck\m_k_\k_M2_gamma_\g.png} &   \includegraphics[width=75mm]{figures/emlkf/emlkf_fdr_nrep100_a_0.2ck\m_k_\k_M2_gamma_\g.png}  \\
(a) Power & (b) FDR 
\end{tabular}
\caption{Number of nonzero groups: \k. }
\end{figure}

    }
}

}

%% file: 08v_details_simulation.tex
We construct the knockoff e-values as in equation (\ref{eq:main_evalue_definition}). In sections \ref{subsec:blockdiagonal} and \ref{subsec:structured_outcomes_simulation}, the (group) knockoff variables are generated via maximum entropy \citep{spector2022powerful, knockoffjulia}. Let $T_j = \beta_j$ and $\widetilde{T}_j = \widetilde{\beta}_j$ where $\beta_j$ and $\widetilde{\beta}_{j}$ are the coefficient estimates based on a 10-fold cross-validated lasso for variable $j$ and its knockoff, respectively \citep{knockoffjulia, candes2018panning}. Following \citet{knockoffjulia}, we compute $T_A = \sum_{j \in A} |T_j|$ and $\widetilde{T}_A = \sum_{j \in A} |\widetilde{T}_j|$ and set $W_A = T_A - \widetilde{T}_A$ for each group $A \in \mathcal{A}$ to get the importance statistics. 

For the chromosome-wide simulation on the UKB genotypes, we use the SHAPEIT knockoff variables by \citet{sesia2021populationstructureshapeit}; see \citet{sesia2021populationstructureshapeit} for further details on pre-processing, importance statistics, knockoff variable generation and population definitions. As in \citet{sesia2021populationstructureshapeit}, to obtain the importance statistics we run a cross-validated lasso including sex (data field 22001-0.0), age (data field 21003-0.0), squared age and the top five genetic principal components as covariates.

We set $c_m = |\mathcal{A}|/|\mathcal{M}|$, our default value, for each $m \in \mathcal{M}$ in all simulations. Following our default recommendations for $\gamma$ discussed in section \ref{sec:kelp}, we set $\gamma = \alpha / 2$ for our simulations in section \ref{subsec:blockdiagonal} and \ref{subsec:structured_outcomes_simulation}. In the more high-dimensional simulation on the UK Biobank genotypes of chromosome 21 in section \ref{subsec:chromosome_simulation_UKB}, we set $\gamma = \alpha / 4$.

Moreover, for the simulations in sections \ref{subsec:blockdiagonal} and \ref{subsec:structured_outcomes_simulation}, following \citet{spector2023controlled}, the absolute value of each nonzero coefficient is set to be at least $0.1 \times \tau$, as otherwise there might be a few nonzero coefficients that are very close to zero. 

Figure \ref{fig:UKB_simulation_snr} shows $\text{Var}(f(X)) / \text{Var}(Y) = \text{Var}(X \beta) / \text{Var}(X \beta + \epsilon)$ as a function of signal amplitude for the simulation setting of section \ref{subsec:chromosome_simulation_UKB}.

\begin{figure}[H]
\centering
\begin{tabular}{c}
 \includegraphics[scale = 0.3]{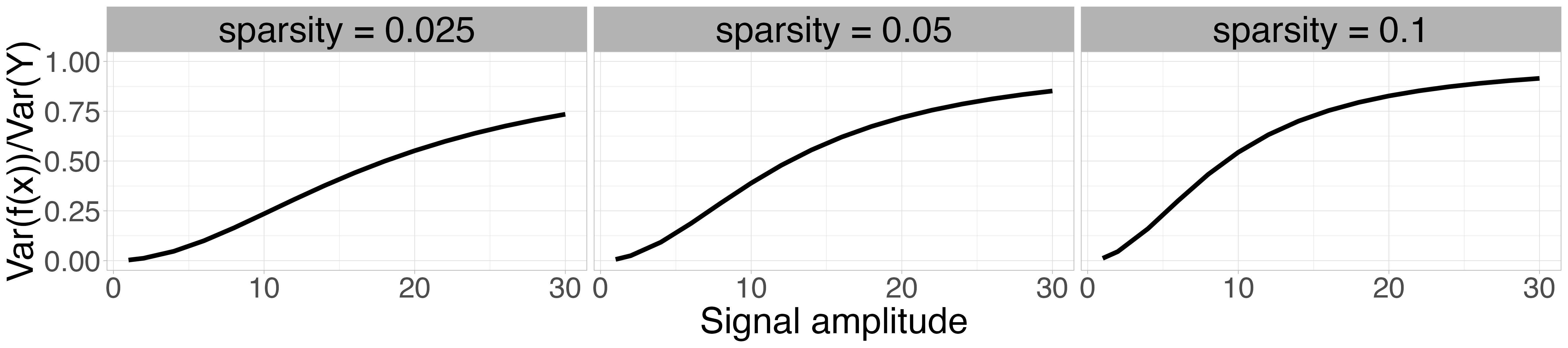} 
\end{tabular}
\caption{$\text{Var}(f(X)) / \text{Var}(Y) = \text{Var}(X \beta) / \text{Var}(X \beta + \epsilon)$ for the chromosome-wide simulation on chromosome 21 of the UK Biobank. Different levels of sparsity by column. Other simulation details as in figure \ref{fig:UKB_simulation} in section \ref{subsec:chromosome_simulation_UKB}.}
\label{fig:UKB_simulation_snr}
\end{figure}

Considering the structured outcome simulations in section \ref{subsec:structured_outcomes_simulation} and the tree depicted in figure \ref{fig:outcome_tree}, for consistency with our previous simulations, we focus on the case where only the rejection of outcome $A$ is reported at the leaf level (but not the rejection of $\{A, B\}$). Specifically, for simplicity, we leverage the connection between the Focused e-BH procedure and KeLP, which is further discussed in appendix \ref{appendix:focusedeBH_vs_kelp}, for our simulations in section \ref{subsec:structured_outcomes_simulation}. Another option is to report both rejections with the argument that while $A$ can be rejected at the leaf level, we can only reject $B$ at the level $\{A, B\}$. It is possible to adjust the location constraint in KeLP (\ref{eq:eLP_location_constraint}) corresponding to either desired definition.

We solve the linear program corresponding to KeLP with CVXR \citep{cvxr} and the solver ``Rglpk'' \citep{Rglpk, glkp}.

%% file: 08vi_details_ukb_application.tex
In all of our UK Biobank applications, we utilize the generated knockoff variables and the partitions by \citet{sesia2021populationstructureshapeit} based on UK Biobank application 27837. The around 592,000 SNPs are partitioned into contiguous groups at seven levels of resolution based on complete-linkage hierarchical clustering. The group sizes range from size one (single SNP) to 425-kb-wide groups; see \citet{sesia2021populationstructureshapeit} for further details on the UKB pre-processing, importance statistic, knockoff variable generation and population definitions. As \citet{sesia2021populationstructureshapeit}, to obtain the importance statistics we run a cross-validated lasso including sex (data field 22001-0.0), age (data field 21003-0.0), squared age and the top five genetic principal components as covariates. Our outcomes of interest are height (data field 50-0.0), platelet count (data field 30080-0.0), platelet crit (data field 30090-0.0), platelet width (data field 30110-0.0) and platelet volume (data field 30100-0.0). 

In this application with seven levels of resolution KeLP runs within minutes using CVXR \citep{cvxr} and the solver ``Rglpk'' \citep{Rglpk, glkp}. For the majority of phenotypes considered the runtime is less than 5 minutes. KeLP for the outcome height has the longest runtime with about 20 minutes.

\subsection{Additional results}

\subsubsection{Height: portion of genome implicated}\label{appendix:height_portion_genome_implicated}

\newpage

\begin{figure}[H]
\begin{tabular}{ >{\centering\arraybackslash} m{1cm} m{10cm} }

1 &  \includegraphics[height=0.1\textwidth, keepaspectratio]{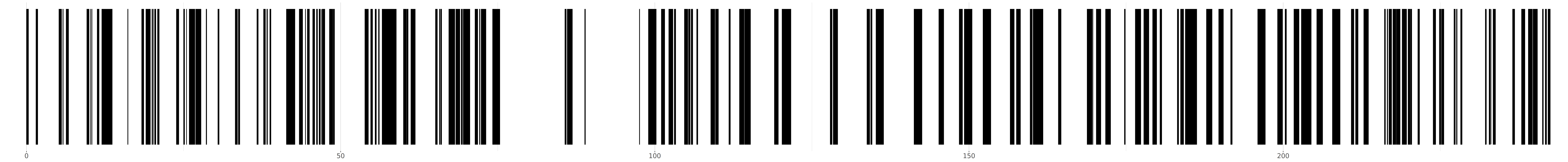}  \\
2 &  \includegraphics[height=0.1\textwidth, keepaspectratio]{figures/UKB/barcodes/height_eblipr_chr_single_line2_gamma_0.0125.png}  \\
3 &  \includegraphics[height=0.1\textwidth, keepaspectratio] {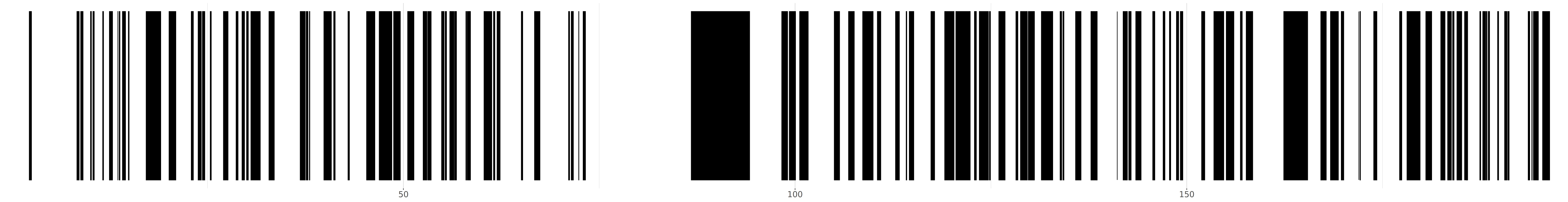}  \\
4 &  \includegraphics[height=0.1\textwidth, keepaspectratio]{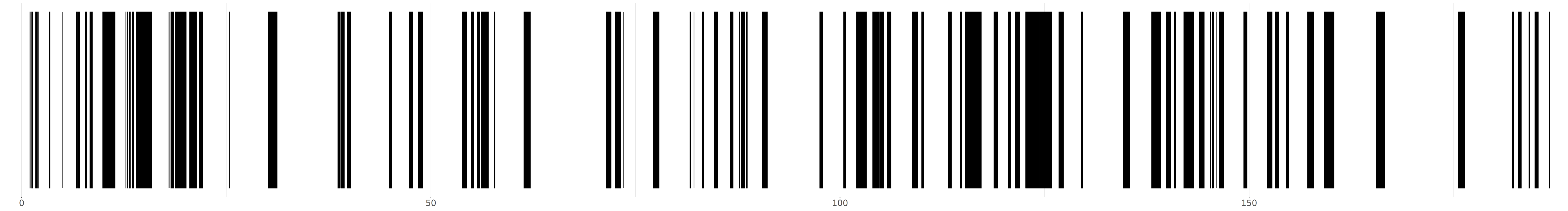}  \\
5 &  \includegraphics[height=0.1\textwidth, keepaspectratio]{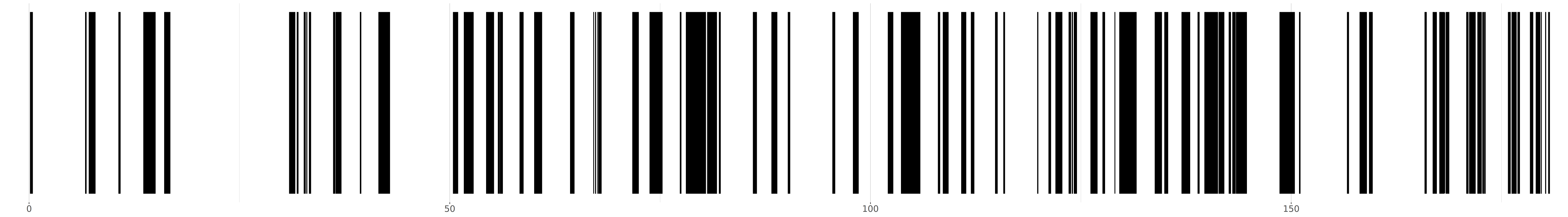}  \\
6 &  \includegraphics[height=0.1\textwidth, keepaspectratio]{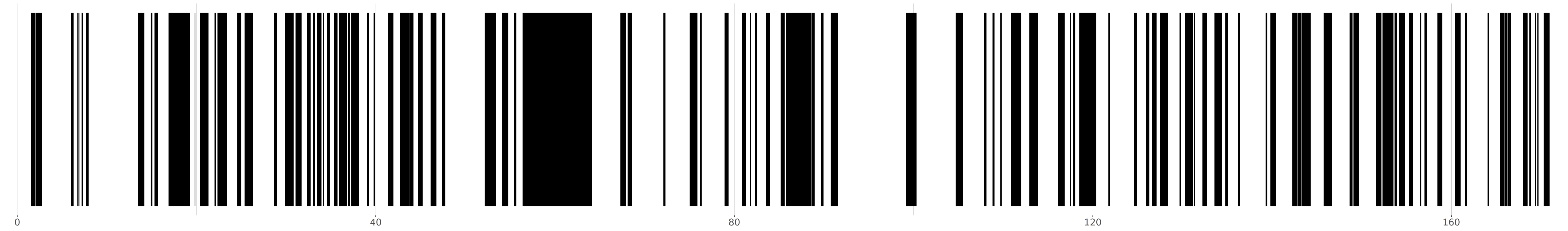}  \\
7 &  \includegraphics[height=0.1\textwidth, keepaspectratio]{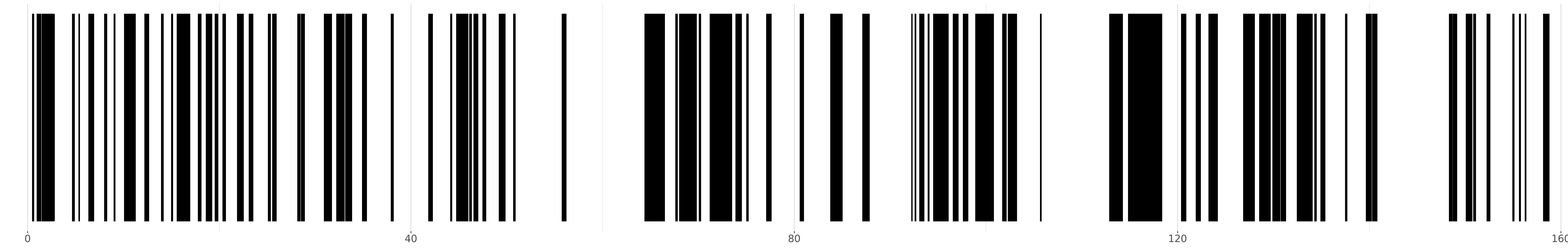}  \\
8 &  \includegraphics[height=0.1\textwidth, keepaspectratio]{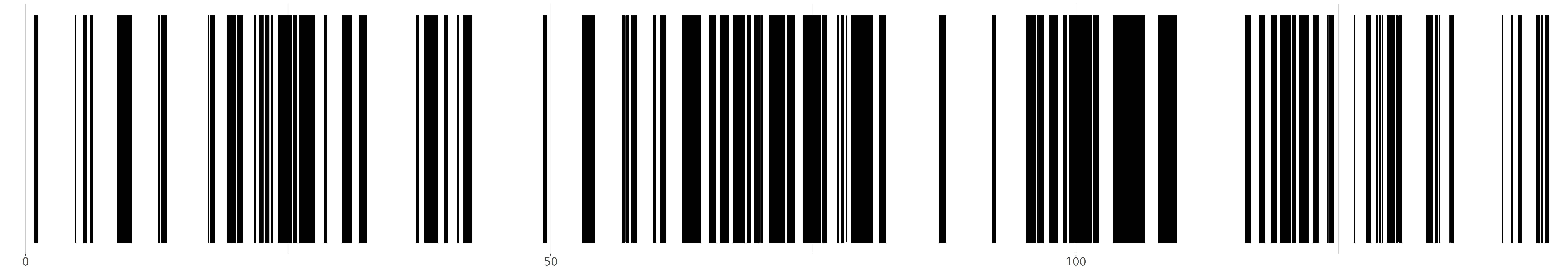}  \\
9 &  \includegraphics[height=0.1\textwidth, keepaspectratio]{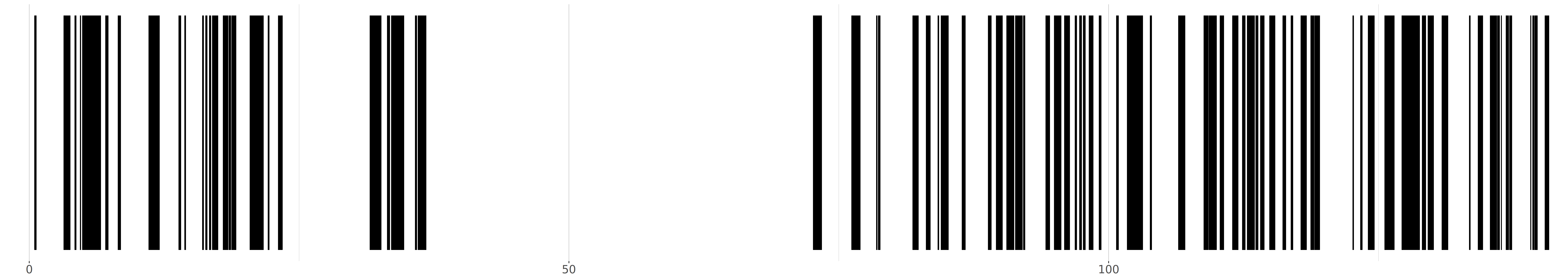}  \\
10 &  \includegraphics[height=0.1\textwidth, keepaspectratio]{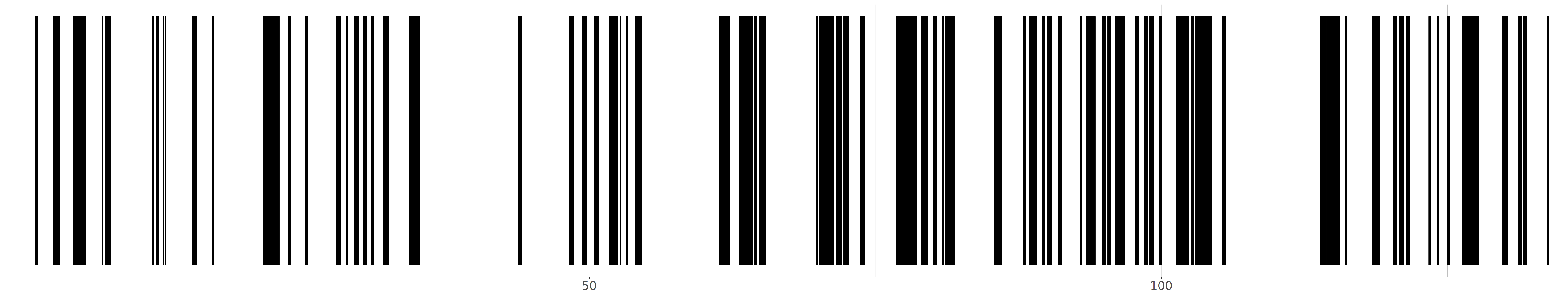}  \\
11 &  \includegraphics[height=0.1\textwidth, keepaspectratio]{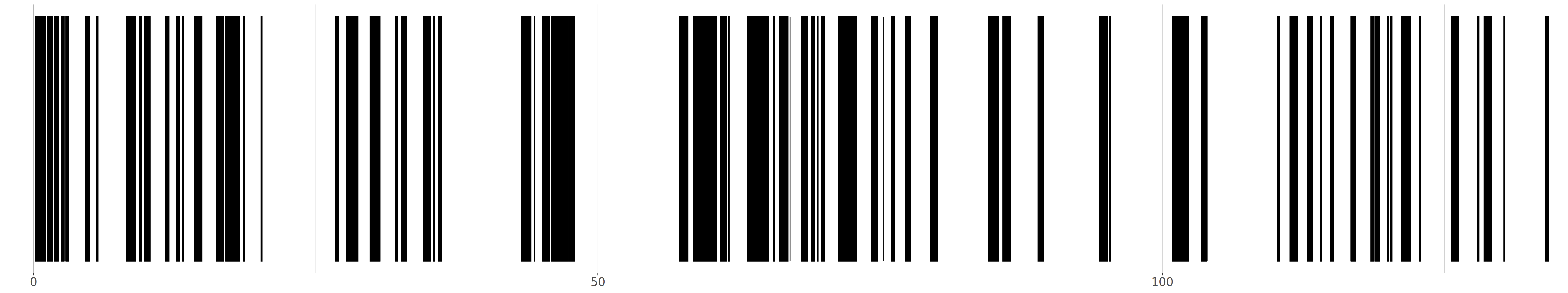} 
\end{tabular}
\caption{Groups rejected by KeLP in each chromosome (1-11): height}
\end{figure}

\newpage

\begin{figure}[H]
\begin{tabular}{ >{\centering\arraybackslash} m{1cm} m{10cm} }

12 &  \includegraphics[height=0.1\textwidth, keepaspectratio]{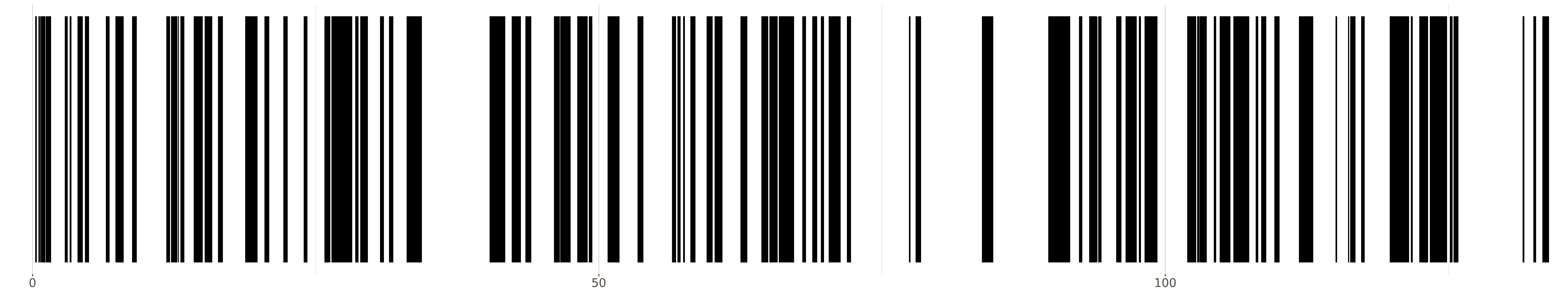}  \\
13 &  \includegraphics[height=0.1\textwidth, keepaspectratio] {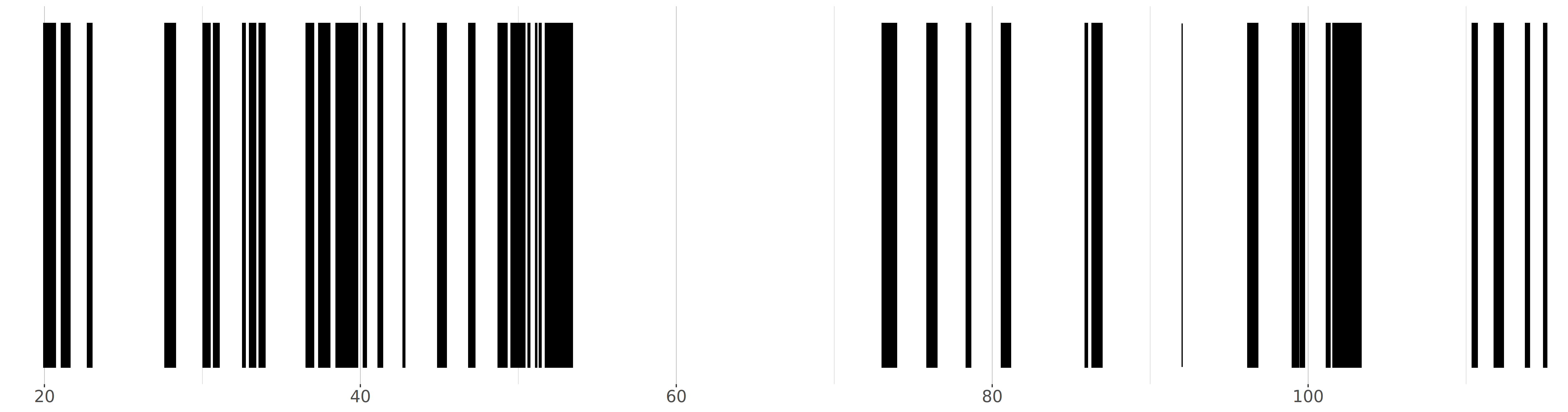}  \\
14 &  \includegraphics[height=0.1\textwidth, keepaspectratio]{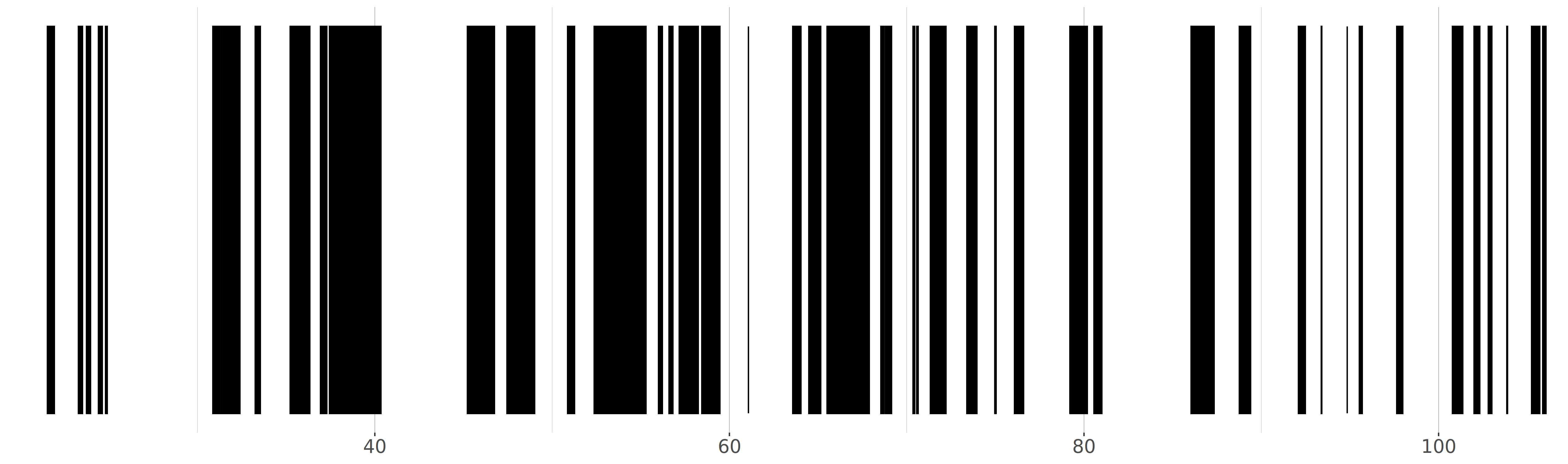}  \\
15 &  \includegraphics[height=0.1\textwidth, keepaspectratio]{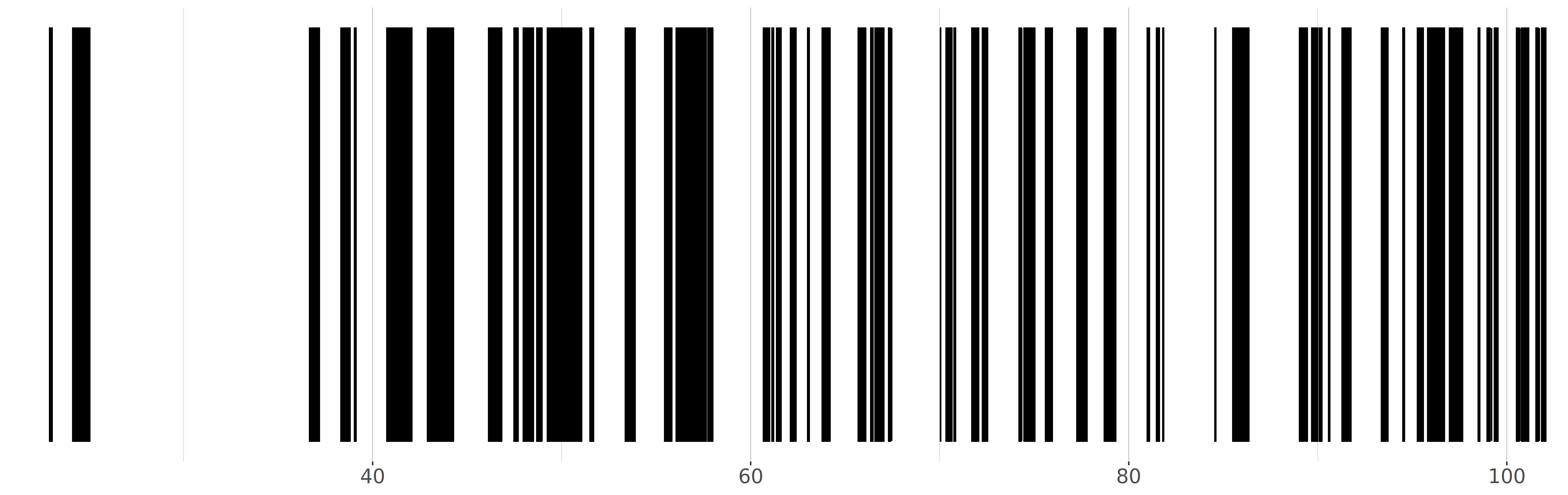}  \\
16 &  \includegraphics[height=0.1\textwidth, keepaspectratio]{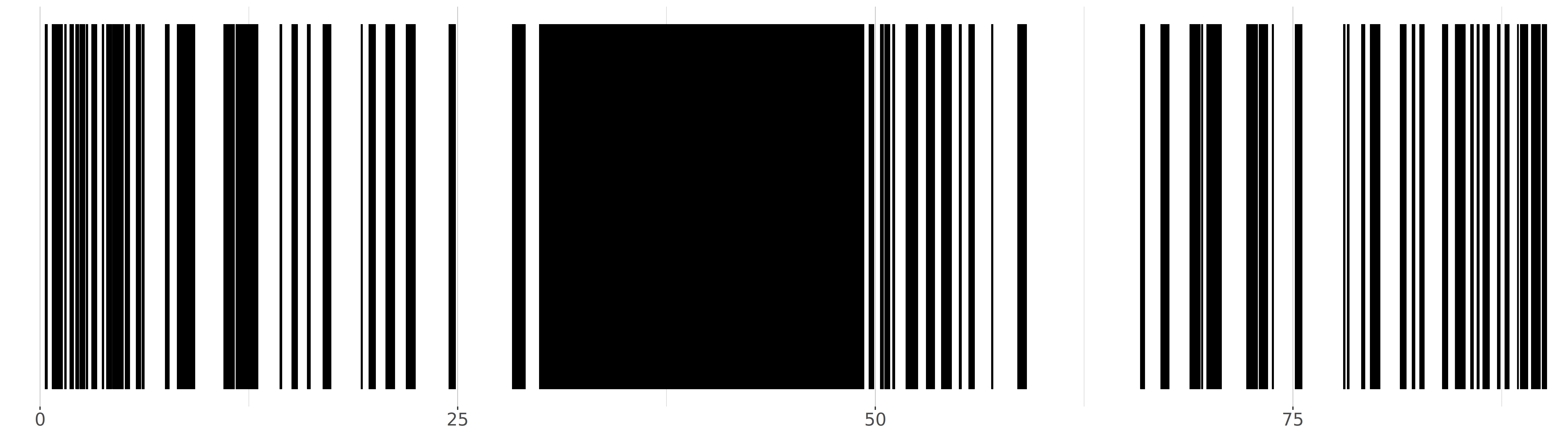}  \\
17 &  \includegraphics[height=0.1\textwidth, keepaspectratio]{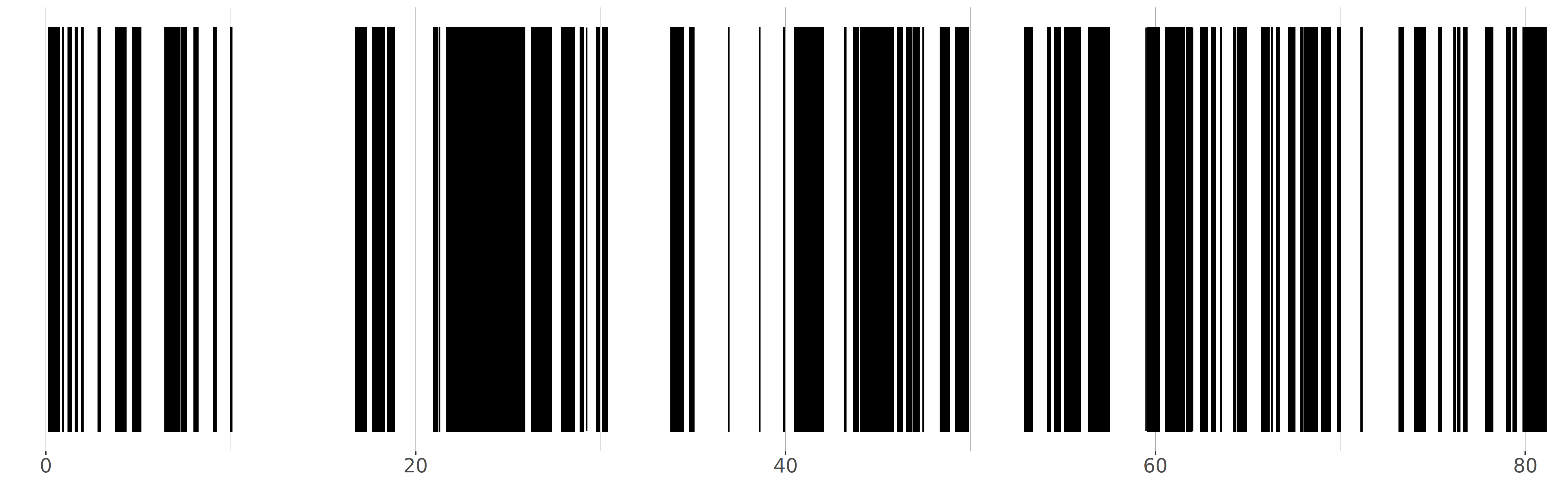}  \\
18 &  \includegraphics[height=0.1\textwidth, keepaspectratio]{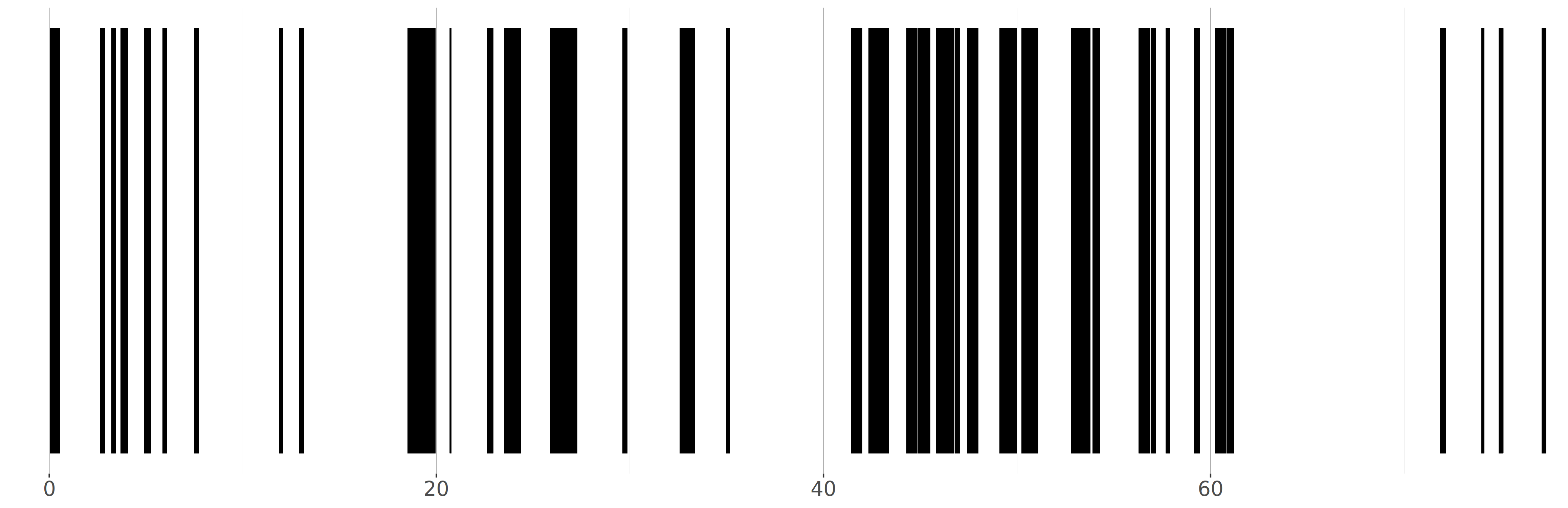}  \\
19 &  \includegraphics[height=0.1\textwidth, keepaspectratio]{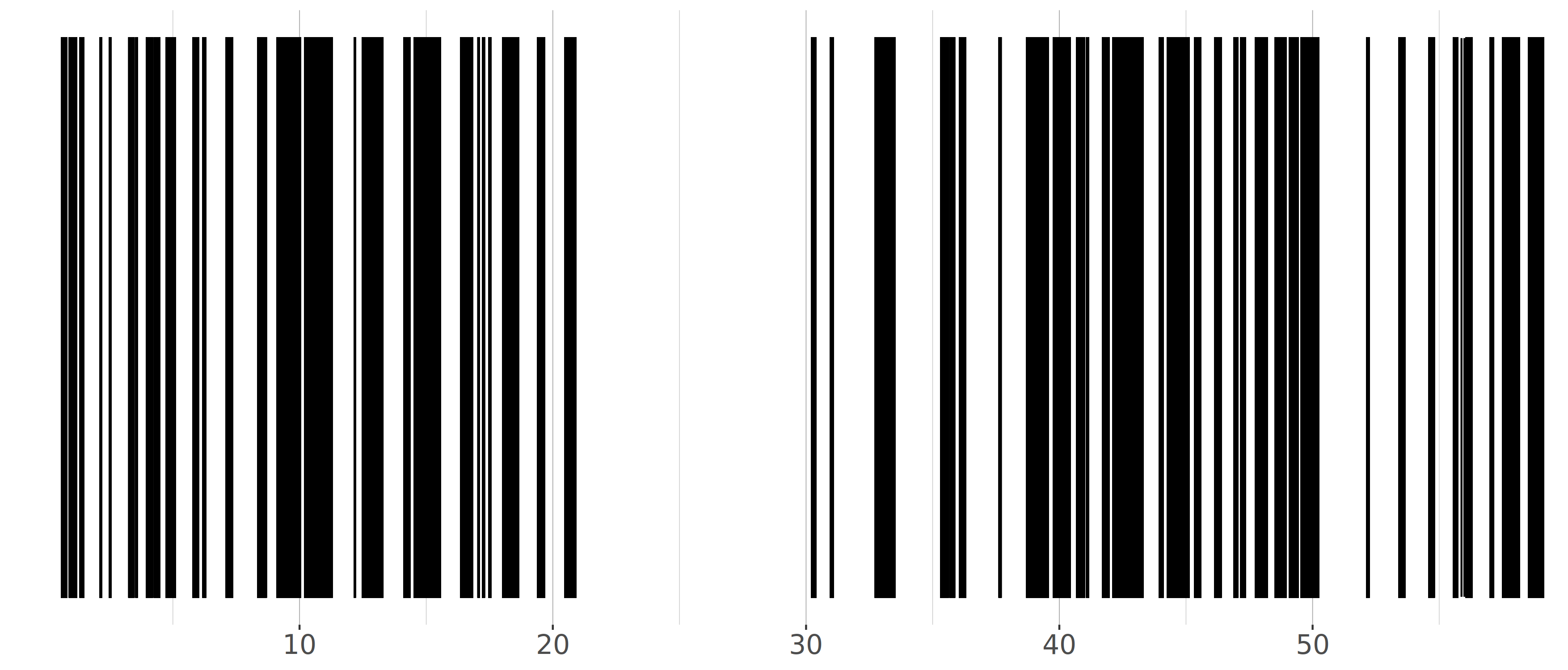}  \\
20 &  \includegraphics[height=0.1\textwidth, keepaspectratio]{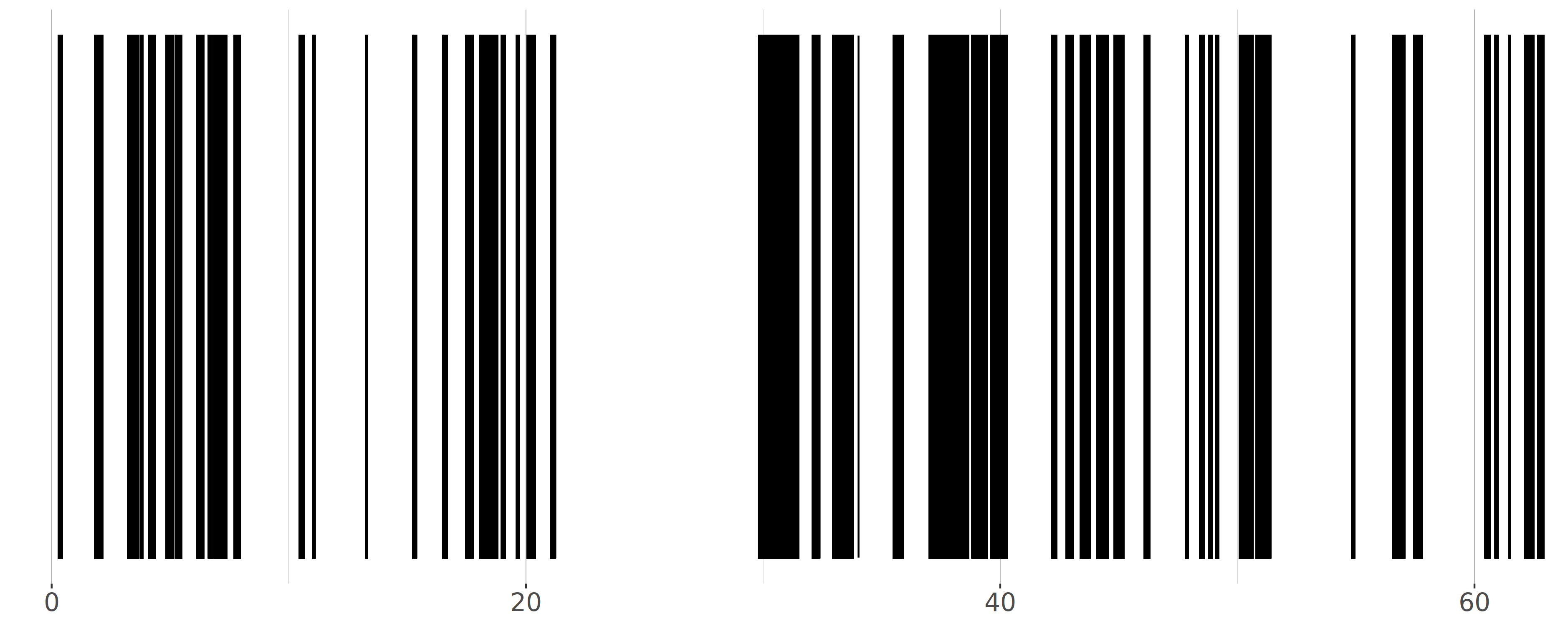}  \\
21 & \includegraphics[height=0.1\textwidth, keepaspectratio]{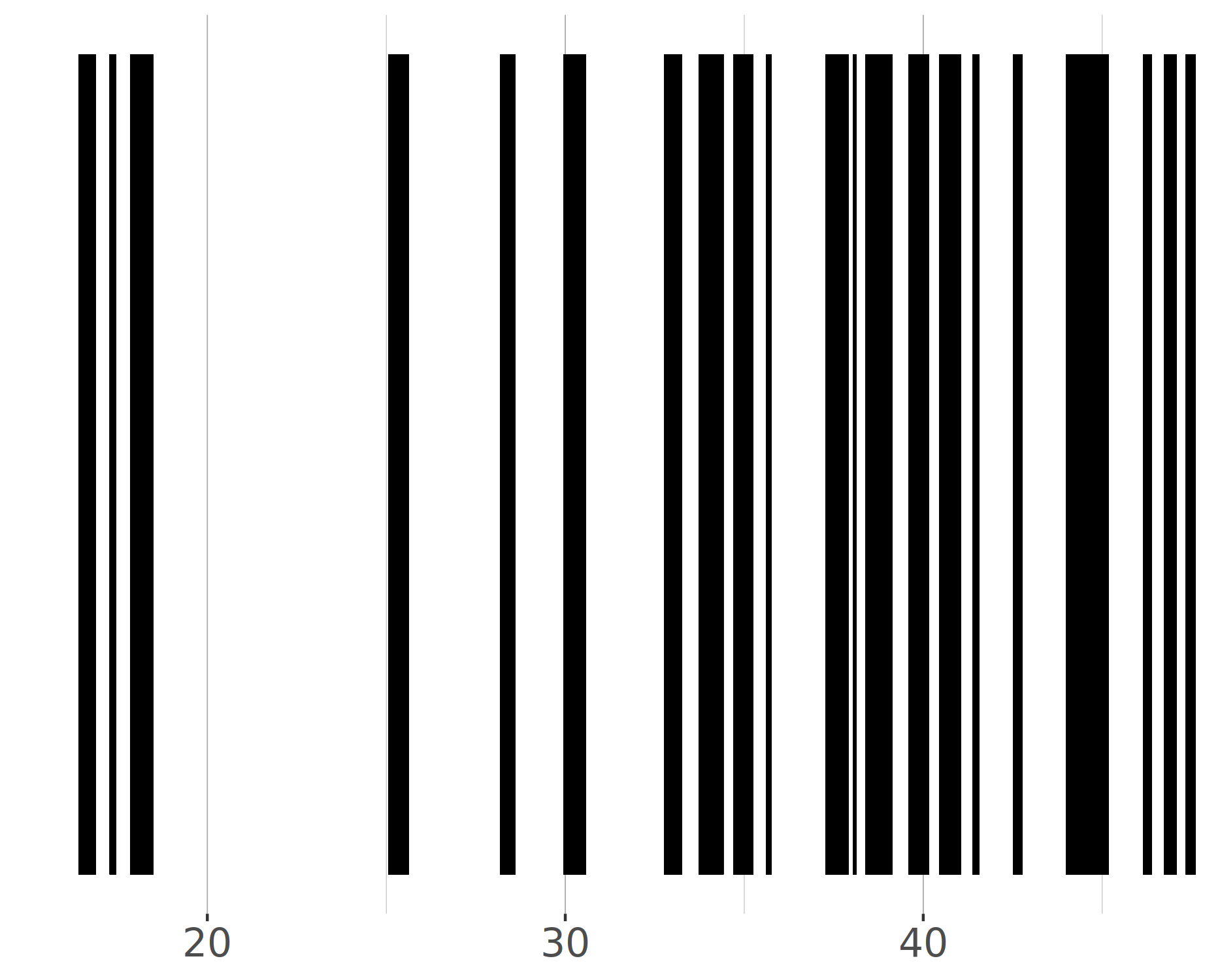}  \\
22 &  \includegraphics[height=0.1\textwidth, keepaspectratio]{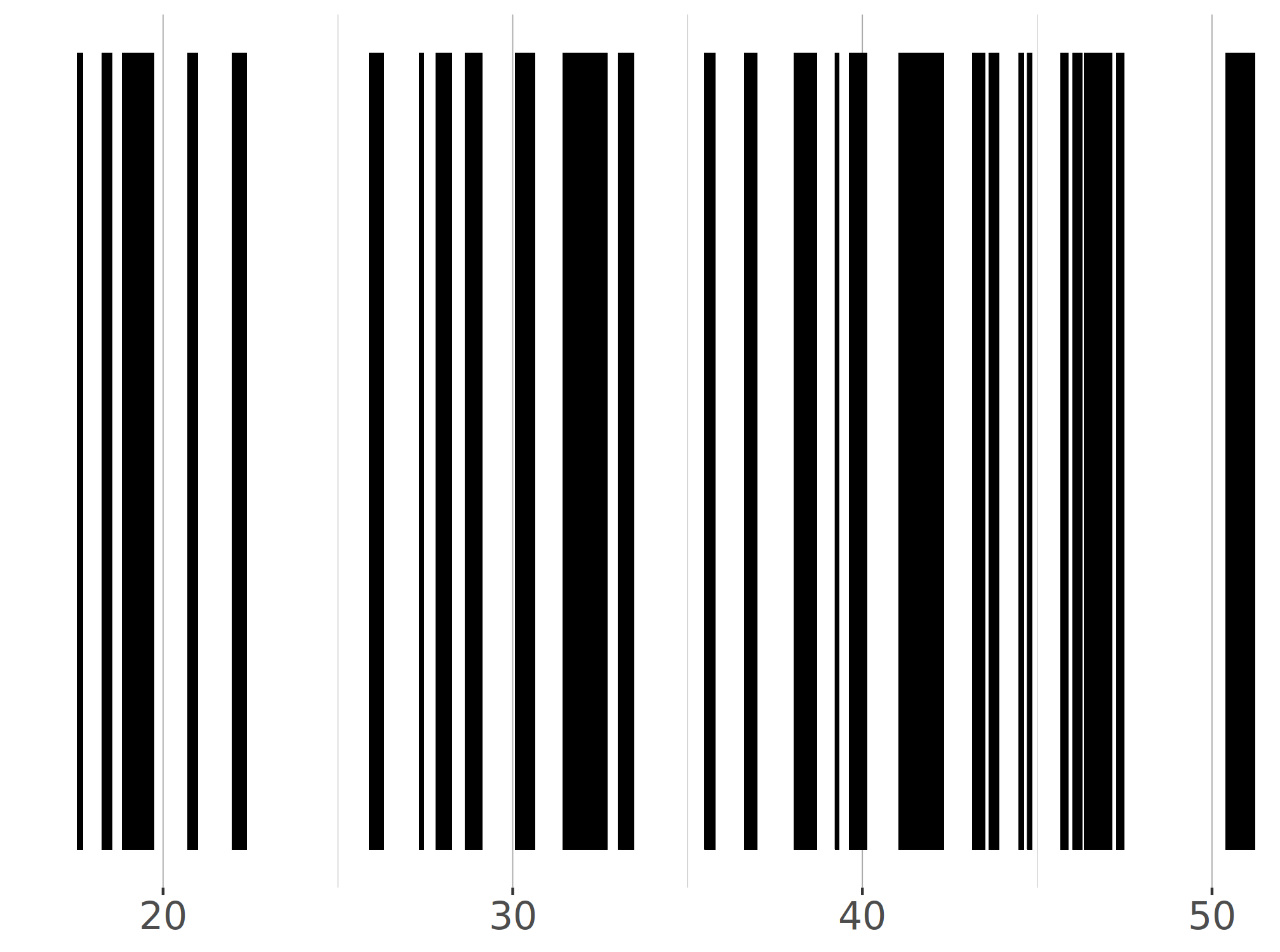}  

\end{tabular}
\caption{Groups rejected by KeLP in each chromosome (12-22): height}
\end{figure}

\newpage

\subsection{Defining e-value parameters for UK Biobank application}

As discussed in section \ref{sec:kelp}, we choose $c_m = \frac{|\mathcal{A}|}{|\mathcal{M}|}$ for each $m \in \mathcal{M}$. Since we are interested in KeLP on the unrelated British population, we use the unrelated White Non-British population as a separate tuning dataset to choose $\gamma$. The definition of populations follows from \citet{sesia2021populationstructureshapeit}. Note that the White Non-British population has a much smaller number of observations compared to the British population, which is our population of interest (14,733 observations compared to 337,117 observations (i.e. less than $5$\% of the British sample size), indicating that a large tuning dataset might not be required.

Figure \ref{fig:smallest_finite_gamma_wnb} displays the smallest $\gamma$ resulting in a finite stopping time $T_m$ on the White Non-British population for our outcomes of interest. These curves for the smallest $\gamma$ displays two ``dips'' at the single-SNP and the 3 kb level. From group sizes starting at 20 kb the smallest $\gamma$ does not exhibit much variation. We therefore choose $\gamma$ to be resolution-dependant. As discussed in section \ref{sec:kelp}, we deterministically set $\gamma = \alpha$ for the single-SNP level and $\gamma = \alpha/2$ on the 3 kb level. For the 20 kb level and up we tune $\gamma$ on the White Non-British population. For simplicity, we will use the same $\gamma$ for each level of resolution starting at 20 kb. For computational efficiency, we run e-BH on a set of pre-filtered groups, where the pre-filtering steps into the direction of KeLP: let a block denote a particular group in the coarsest level of resolution. Within each block, we filter to the group with the largest weight-adjusted fraction (where the weight is based on inverse group-size). We then run e-BH on this pre-filtered set of groups and choose the $\gamma$ which results in the largest number of non-zero rejections across different $\alpha$. If there are ties, we choose the larger $\gamma$.

\begin{figure}[H]
\centering
  \includegraphics[width=0.55\textwidth]{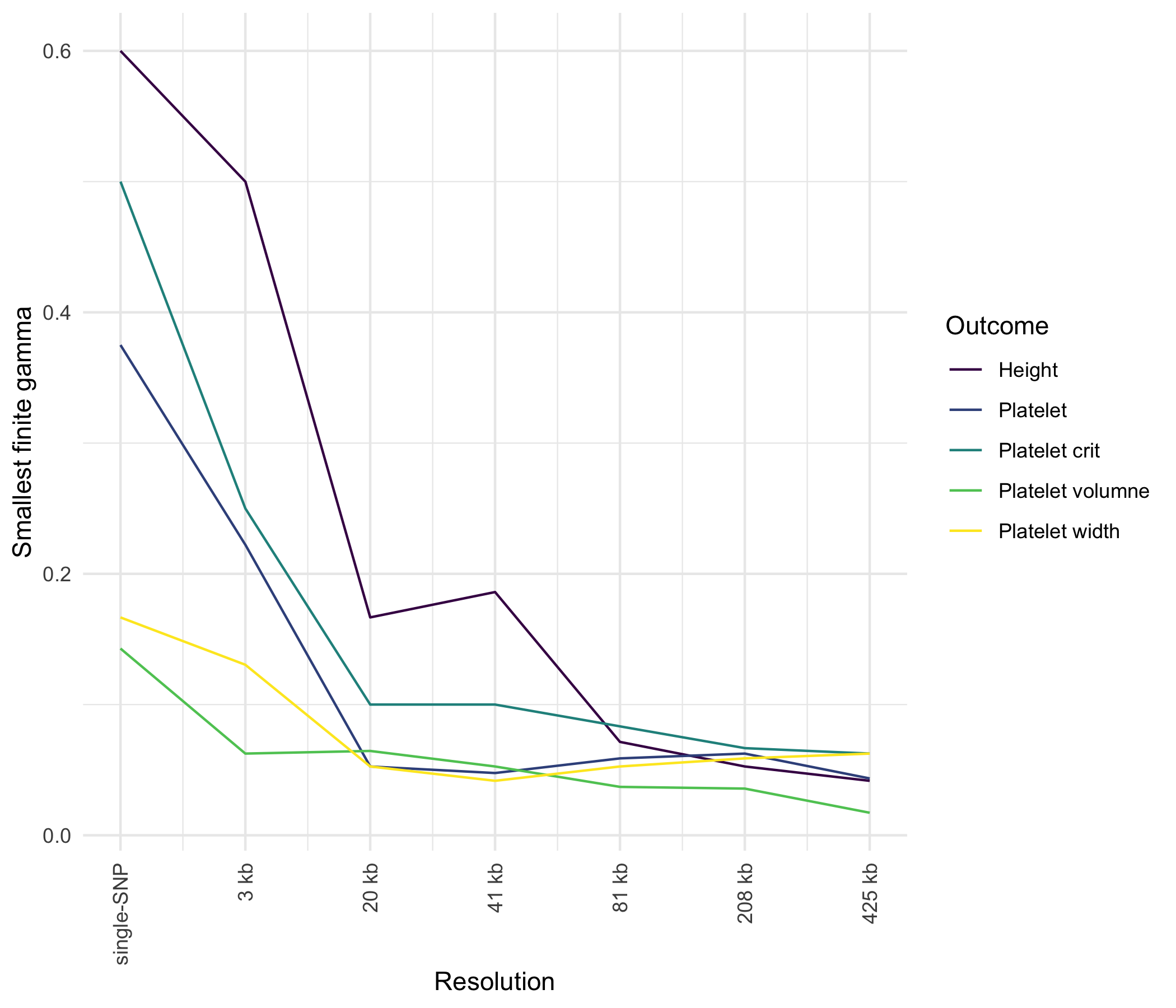}
\caption{Smallest finite $\gamma$ for White Non-British population by resolution.}
\label{fig:smallest_finite_gamma_wnb}
\end{figure}